\documentclass[preprint,10pt,authoryear]{elsarticle}

\usepackage{amssymb}
\usepackage{amsthm}
\usepackage{natbib}
\usepackage{amsmath}
\usepackage{bm}
\usepackage{lineno}
\usepackage{hyperref}
\usepackage{caption}
\usepackage{subcaption}
\usepackage{multirow}
\usepackage{placeins}
\usepackage{chngcntr}
\usepackage{lineno}
\usepackage{float}
\usepackage{makecell}
\usepackage{rotating}

\newtheorem{prop}{Proposition}
\usepackage[a4paper, total={6.5in, 9in}]{geometry}

\usepackage{todonotes}

\usepackage{xcolor}

\usepackage[inline,shortlabels]{enumitem}
\usepackage{ctable}

\DeclareMathAlphabet\mathbfcal{OMS}{cmsy}{b}{n}

\newcommand{\bx}{\boldsymbol{x}}
\newcommand{\bu}{\boldsymbol{u}}
\newcommand{\bv}{\boldsymbol{v}}
\newcommand{\bz}{\boldsymbol{z}}
\newcommand{\bX}{\boldsymbol{X}}
\newcommand{\bU}{\boldsymbol{U}}
\newcommand{\bV}{\boldsymbol{V}}
\newcommand{\bZ}{\boldsymbol{Z}}

\newcommand{\bw}{\boldsymbol{w}}

\newcommand{\bbeta}{\boldsymbol{\beta}}
\newcommand{\balpha}{\boldsymbol{\alpha}}
\newcommand{\bgamma}{\boldsymbol{\gamma}}
\newcommand{\blambda}{\boldsymbol{\lambda}}
\newcommand{\logit}{\mathrm{logit}}

\usepackage{amssymb}

\journal{arXiv}

\begin{document}

\begin{frontmatter}



\title{Modeling Bounded Count Environmental Data Using a Contaminated Beta-Binomial Regression Model}


\author[inst1]{Otto, Arnoldus F. \corref{cor1}}
\ead{arno.otto@up.ac.za}
\cortext[cor1]{Corresponding author}
\author[inst2]{Punzo, Antonio}
\author[inst1,inst3]{Ferreira, Johannes T.}
\author[inst1,inst4]{Bekker, Andri\"ette}
\author[inst2]{Tomarchio, Salvatore D.}
\author[inst5]{Tortora, Cristina}

\affiliation[inst1]{organization={Department of Statistics, University of Pretoria, Pretoria, South Africa}}
\affiliation[inst2]{organization={Department of Economics and Business, University of Catania, Catania, Italy}}
\affiliation[inst3]{organization={School of Statistics and Actuarial Science, University of the Witwatersrand, Johannesburg, South Africa}}
\affiliation[inst4]{organization={Centre for Environmental Studies, Department of Geography, Geoinformatics and Meteorology, University of Pretoria, Pretoria, South Africa}}
\affiliation[inst5]{organization={Department of Mathematics and Statistics, San José State University, California, United States of America}}

\begin{abstract}
Climate change is a crucial aspect of environmental challenges, with profound implications for human well-being, affecting vital ecosystem services such as clean water, food production, and pollination. 
It also causes species displacement, habitat loss, and increased extinction risks. 
This paper investigates two environmental applications related to climate change, where observations consist of bounded counts.
The first application examines the effect of winter malnutrition on mule deer (\textit{Odocoileus hemionus}) fawn mortality. 
Animals are crucial in ecosystem services, including seed dissemination, pollination, and pest control. 
Since climate change affects species' survival, it is essential to study these impacts further.
The second application investigates the public perception of climate change, which is vital for shaping effective policies and environmental strategies. The Eurobarometer 95.1 survey, conducted in March–April 2021, included a key question assessing the perceived severity of climate change on a scale from 1 to 10. Analyzing responses to this question is crucial for gaining deeper insights into public sentiment regarding climate change.
The binomial and beta-binomial (BB) models are commonly used for bounded count data, with the BB model offering the advantage of accounting for potential overdispersion. However, extreme observations in real-world applications may hinder the performance of the BB model and lead to misleading inferences.
To address this issue, we propose the contaminated beta-binomial (cBB) distribution (cBB-D), which provides the necessary flexibility to accommodate extreme observations. The cBB model accounts for overdispersion and extreme values while maintaining the mean and variance properties of the BB distribution.
The availability of covariates that improve inference about the mean of the bounded count variable motivates the further proposal of the cBB regression model (cBB-RM). 
Different versions of the cBB-RM model — where none, some, or all of the cBB parameters are regressed on available covariates — are fitted to the datasets. 
The effectiveness of our model is also demonstrated through a sensitivity analysis to assess the impact of extreme values on parameter estimation.
\end{abstract}


\end{frontmatter}

\section{Introduction} 
\label{sec:Intro}

In environmental research, data often consists of bounded counts \citep{paul2007generalized,ryan2007application}. This is particularly noteworthy in climate change studies, which represent one of our most pressing global challenges, with far-reaching implications for ecosystems \citep{yee2008comparing}, and economies \citep{sciandra2024discrete}.

This paper presents two environmental applications related to climate change, with the first focusing on the impact of winter malnutrition on mule deer (\textit{Odocoileus hemionus}) fawn mortality. Climate-driven shifts in environmental conditions can influence species survival and ecosystem dynamics, making the study of such impacts increasingly important \citep{aspinall1994climate,muluneh2021impact}. Mule deer inhabit diverse environments, and their populations vary in response to numerous factors. Most noticeable changes are those that occur immediately and dramatically, such as a decline following a harsh winter. Populations may also decline or increase for extended periods and broader geographic range. This data was gathered from 1875 radio-collared fawns captured in early winter across Colorado, Idaho, and Montana in the United States of America from 1981 to 1996 \citep{unsworth1999mule}. This dataset reflects an extensive investigation into the overwinter survival of a species that is particularly responsive to environmental factors due to climate change.

The second application is understanding the public perception of climate change, which is crucial for policymakers and environmental organizations striving to implement effective climate strategies. Public attitudes towards climate change can significantly influence the adoption of sustainable practices and the political will to enact necessary regulations \citep{arikan2021public}.
Eurobarometer is a series of public opinion surveys conducted regularly on behalf of the European Commission and other European Union (EU) institutions since 1973. 
These surveys address various topical issues relating to the EU throughout its member states, including their opinion on climate-related topics. 
In the Eurobarometer 95.1 survey, conducted between March and April 2021, one key question posed was:
\begin{quote}
\emph{``How serious a problem do you think climate change is at this moment? Please use a scale from 1 to 10, with `1' meaning it is `not at all a serious problem' and `10' meaning it is `an extremely serious problem.'"}
\end{quote}
This question provides a quantitative measure of the perceived severity of climate change, reflecting individual concerns and priorities regarding this global issue. The ordinal nature of the response scale, ranging from 1 to 10, allows for a nuanced analysis of the degree to which climate change is perceived as a threat. The availability of personal information allows us to analyze the data about various demographic and socio-economic factors further. 
Surveys such as these include many important variables expressed as bounded counts.

Let $Y$ be the bounded count variable of interest, taking values in $\left\{0,1,\ldots,m\right\}$, where $m\geq 1$ is known in advance.
From a statistical perspective, $Y$ is often regarded as the count of successes out of $m$ trials. 
For the following discussion, it may be helpful to consider both perspectives on this variable.
Handling sample data regarding $Y$ presents a nuanced challenge.
While reducing such data into proportions may be tempting, it is generally preferable to model the raw counts rather than convert them to a proportion before modelling. 
This approach allows for the weighting of larger trials more than smaller trials rather than treating all proportions equally. 
For example, a proportion of $0.5$ could represent various scenarios, such as a single success out of two attempts or four successes out of four attempts. 
Solely focusing on proportions obscures the intricate context underlying the data and imposes limitations on capturing the true variability inherent in count-based observations \citep{martin2020modeling}.

The one-parameter binomial distribution (B-D) and its extension, the binomial regression model (B-RM), have traditionally been the standard choices as they directly model the raw bounded counts.  
The B-D is defined by the following probability mass function (PMF):
\begin{equation}
    f_{\text{B}_m}(y;\pi)={m\choose y}\pi^y(1-\pi)^{m-y}, \quad y=0,1,\dots,m,
    \label{eq:Binom PMF}
\end{equation}
where $\pi \in (0,1)$ represents the probability of success in each trial. The notation reflects the fact that in the considered context, $\pi$ is the only unknown parameter, whereas $m$ is an inherent characteristic of the phenomenon under study.
If $Y$ has the PMF in \eqref{eq:Binom PMF}, we denote it as $Y\sim \mathcal{B}_m\left(\pi\right)$. 
The moments and characteristics of practical interest of $Y\sim \mathcal{B}_m\left(\pi\right)$, namely mean, variance, skewness, and excess kurtosis, are:  
\begin{equation*}
    \mathrm{E}_{\text{B}_m}(Y;\pi)=m\pi,
\end{equation*}
\begin{equation}
    \mathrm{Var}_{\text{B}_m}(Y;\pi)=m\pi(1-\pi),\label{eq:B var}
\end{equation}
\begin{equation}
    \mathrm{Skew}_{\text{B}_m}(Y;\pi)=\frac{1-2\pi}{\sqrt{m\pi(1-\pi)}}, \label{eq:B skew}
\end{equation}
and  
\begin{equation}
    \mathrm{ExKurt}_{\text{B}_m}(Y;\pi)=\frac{1-6\pi(1-\pi)}{m\pi(1-\pi)} \label{eq:ExKurt B}.
\end{equation}  
All these characteristics are governed by the single parameter $\pi$.
The variance in \eqref{eq:B var} lies in $(0,m/4]$, attaining its maximum at $\pi = 0.5$.
The skewness in \eqref{eq:B skew} depends on $\pi$, being positive for $\pi < 0.5$, negative for $\pi > 0.5$, and zero when $\pi = 0.5$; it is always bounded within the interval $\left(-1/\sqrt{m\pi(1-\pi)}, 1/\sqrt{m\pi(1-\pi)}\right)$. 
As $m$ increases, the skewness decreases in magnitude, approaching zero. 
The excess kurtosis in \eqref{eq:ExKurt B} reaches its minimum at $\pi = 0.5$, taking value $-2/m$, and diverges to infinity as $\pi$ approaches 0 or 1.
Consequently, as $m$ increases, the range of possible negative excess kurtosis values narrows. 
Although the characteristics may vary, their dependence solely on $\pi$ restricts the binomial model's ability to capture the empirical behavior of bounded counts accurately. 
In particular, it is well known that the binomial model is inadequate in the presence of overdispersion, where the observed variance exceeds the binomial variance for a given mean. 

The beta-binomial (BB) distribution (BB-D) has emerged as a natural extension of the B-D to address the latter limitation, and plays the same role for the B-D as the negative binomial distribution does for the Poisson distribution.  
Its flexibility makes it a powerful tool in various fields, including epidemiology \citep{arostegui2010assessment} and microbiology \citep{martin2020modeling}, 
where the underlying success probabilities (i.e., $\pi$) may vary or be uncertain. 

The mean-parameterized or, more appropriately, the $\pi$-parametrized BB-D has the following PMF:
\begin{equation} 
\label{pmf mean beta binomial distribution}
   f_{\text{BB}_m}(y;\pi,\sigma) = {m\choose y}\frac{\displaystyle\mathrm{B}\left(y+\frac{\pi}{\sigma},m-y+\frac{1-\pi}{\sigma}\right)}{\displaystyle\mathrm{B}\left(\frac{\pi}{\sigma},\frac{1-\pi}{\sigma}\right)}, \quad y=0,1,\dots,m, 
\end{equation}
where $\pi\in(0,1)$ is the analogue of the binomial probability of success parameter, $\sigma>0$ is the dispersion parameter, and $\mathrm{B}(\cdot,\cdot)$ denotes the beta function. 
If $Y$ has the PMF in \eqref{pmf mean beta binomial distribution}, then we simply write $Y\sim\mathcal{BB}_m(\pi,\sigma)$. 
The expected value and variance of $Y\sim\mathcal{BB}_m(\pi,\sigma)$ are
\begin{equation}
    \mathrm{E}_{\text{BB}_m}(Y;\pi)=m\pi, \label{ec:BB mean}
\end{equation}
as for the B-D, and
\begin{align}
\label{eq variance betabinomial}
\mathrm{Var}_{\text{BB}_m}(Y;\pi,\sigma)&=m\pi(1-\pi)\left[1+(m-1)\frac{\sigma}{1+\sigma}\right]\nonumber\\
&=\mathrm{Var}_{\text{B}_m}(Y;\pi)\left[1+(m-1)\frac{\sigma}{1+\sigma}\right].
\end{align}
Consequently, the mean in \eqref{ec:BB mean} is solely determined by $\pi$, whereas the variance in \eqref{eq variance betabinomial} is influenced by both $\pi$ and $\sigma$.
The parameterization of the PMF in \eqref{pmf mean beta binomial distribution} in terms of $\pi$ and $\sigma$ offers enhanced statistical interpretability and, because of \eqref{ec:BB mean}, allows the use of the BB-D in a regression (toward the mean) context. 
Since the term $\frac{\sigma}{1+\sigma}$ in \eqref{eq variance betabinomial} lies within $(0,1)$, the term on squared brackets on the right-hand side of \eqref{eq variance betabinomial} can assume values in $(1,m)$ and, as such, it acts as an inflation factor that allows the BB-variance in \eqref{eq variance betabinomial} to span from $\mathrm{Var}_{\text{B}_m}(Y;\pi)$ to $m\mathrm{Var}_{\text{B}_m}(Y;\pi)$, thereby accommodating varying degrees of binomial overdispersion. 
This is why one can interpret $\sigma$ as a measure of the overdispersion.
Another aspect which is important to emphasize is that, as $\sigma$ increases in $\mathcal{BB}_m(\pi,\sigma)$, the variance in \eqref{eq variance betabinomial} also increases, making observations at the extremes of the support $\left\{0,1,\ldots,m\right\}$ -- which we will frequently refer to as extreme observations in this paper -- more likely compared to the binomial distribution $\mathcal{B}_m(\pi)$.
As a limiting case, in the scenario of maximum variance ($\sigma\to\infty$), the BB-D tends to a two-point distribution assuming values $0$ and $m$ with probabilities $1-\pi$ and $\pi$, respectively.
Some researchers favor the reparameterization in terms of $\pi$ and $\gamma=\frac{\sigma}{1+\sigma}$; see, e.g., \citet{bayes2024robust}.  
However, since $\sigma$ and $\gamma$ share a one-to-one correspondence, they convey equivalent information about the BB-D.
Finally, skewness and excess kurtosis for $Y\sim\mathcal{BB}_m(\pi,\sigma)$ are given by 
\begin{align}\label{eq BB skew}
\text{Skew}_{\text{BB}_m}(Y;\pi,\sigma) &=\frac{(1-2 \pi) (2 m \sigma +1) }{(2 \sigma +1)\sqrt{\frac{ m \pi (1-\pi)  (m \sigma +1)}{\sigma +1}}}\notag\\
    &= \frac{(1-2 \pi) (2 m \sigma +1)}{(2 \sigma +1)\sqrt{\text{Var}_{\text{BB}_m}(Y;\pi,\sigma)}},
\end{align}
and
\begin{align}\label{eq BB kurt}
\text{ExKurt}_{\text{BB}_m}(Y;\pi,\sigma)=\frac{6\pi(1-\pi)(m(6\sigma+5)\sigma(m\sigma+1)+\sigma+1)-(\sigma+1)(\sigma(6m(m\sigma+1)-1)+1)}{\pi(1-\pi) m(2\sigma+1)(3\sigma+1)(m\sigma+1)}.
\end{align}
The skewness and excess kurtosis of the $\mathcal{BB}_m(\pi,\sigma)$ are governed by both $\pi$ and $\sigma$.  
The skewness in \eqref{eq BB skew} is always bounded within the interval 
$$
\left(-\frac{(2 m \sigma +1) }{(2 \sigma +1)\sqrt{\frac{(1-\pi) m \pi   (m \sigma +1)}{\sigma +1}}},\frac{(2 m \sigma +1) }{(2 \sigma +1)\sqrt{\frac{(1-\pi) m \pi   (m \sigma +1)}{\sigma +1}}}\right).
$$ 
As for the B-D case, it is zero when $\pi=0.5$, positive for $\pi<0.5$, and negative for $\pi>0.5$. Increasing $\sigma$ leads to an increase in the magnitude of the skewness, while an increase in $m$ decreases it. The excess kurtosis in \eqref{eq BB kurt} reaches its minimum $-\frac{2(3m\sigma(m\sigma+1)+\sigma+1)}{m(1+3\sigma)(1+m\sigma)}$ at $\pi = 0.5$, and diverges to infinity as $\pi$ approaches $0^+$ or $1^-$.
The excess kurtosis also increases as $\sigma$ increases, except between the range
\begin{align*}
   \pi \in &\left(\frac{1}{2}-\frac{\sqrt{3}}{6} \sqrt{\frac{\left(2 \sigma _1+1\right) \left(2 \sigma _2+1\right) \left(3 m \left(\sigma _2 \left(m \sigma _1+\sigma _1+1\right)+\sigma _1\right)+2\right)}{6 m \sigma _2^2 \left(m \sigma _1+\sigma _1+1\right)+\sigma _2 \left(m \sigma _1+\sigma _1+1\right) \left(m \left(6 \sigma _1+5\right)+6\right)+\sigma _1 \left(m \left(6 \sigma _1+5\right)+6\right)+4}},\right.\\
    &\left.\frac{1}{2}+\frac{\sqrt{3}}{6} \sqrt{\frac{\left(2 \sigma _1+1\right) \left(2 \sigma _2+1\right) \left(3 m \left(\sigma _2 \left(m \sigma _1+\sigma _1+1\right)+\sigma _1\right)+2\right)}{6 m \sigma _2^2 \left(m \sigma _1+\sigma _1+1\right)+\sigma _2 \left(m \sigma _1+\sigma _1+1\right) \left(m \left(6 \sigma _1+5\right)+6\right)+\sigma _1 \left(m \left(6 \sigma _1+5\right)+6\right)+4}}
    \right),
\end{align*}
where the excess kurtosis for $\sigma_2>\sigma_1>0$ is smaller for $\sigma_2$ than for $\sigma_1$. 
For example, if $m=10,\sigma_1=0.1$ and $\sigma_2=1$, the excess kurtosis for $\sigma_2$ is larger than that for $\sigma_1$ when $\pi \in (0,0.223)$ and $\pi \in (0.777,1)$, smaller when $\pi \in(0.223,0.777)$, and equal when $\pi =0.223$ or $\pi=0.777$.


As noted by \citet{bayes2024robust}, despite the ability of the BB-D to capture binomial overdispersion, it is limited by its two parameters, $\pi$ and $\sigma$, which do not offer sufficient flexibility to fully capture higher-order characteristics such as skewness and excess kurtosis (refer to \eqref{eq BB skew} and \eqref{eq BB kurt}). 
The latter is particularly important as it may indicate an excess of extreme observations relative to the fitted BB-D. 
Consequently, the BB-D may struggle to model empirical skewness and kurtosis accurately, as it cannot fully adjust the distribution's shape beyond the mean and variance. 
From an inferential perspective, fitting the BB-D in the presence of extreme observations can lead to underestimation of standard errors and overstatement of the significance of estimates, resulting in potentially misleading inferences.

Motivated by these considerations, we introduce the contaminated BB (cBB) distribution (cBB-D). 
This model is formulated as a two-component mixture of BB distributions, where one component represents the typical (or regular) observations (the reference BB-D) and the other, with the same mean but an increased dispersion accounts for the excess of extreme observations (the contaminant distribution). 
The underlying approach mirrors that of other studies that utilize contaminated distributions to address varying data supports; see, for example, \citet{Punz:Anew:2019}, \citet{Punz:Bagn:PhyA:2021,punzo2024asymmetric}, \citet{otto2025contaminated}, and \citet{tomarchio2024new}. 
For a detailed discussion of the reference distribution concept, which we assume to be the BB-D in this case, refer to \citet{davies1993identification}.

The proposed cBB-D offers a flexible method for fitting a regression model to bounded count outcomes, particularly when overdispersion and an excess of extreme observations are present. 
Additionally, our model has the advantage of a closed-form PMF and interpretable parameters. 
As highlighted by \citet{ley2021flexible}, \citet{otto2024refreshing}, and \citet{wagener2024uncovering}, it is crucial for parameters to have meaningful interpretations to draw valid inferences about the underlying population from which the data originate. 
In detail, in addition to the standard BB-D parameters, the cBB-D introduces two new parameters: one representing the proportion of observations from the contaminant BB-D and another indicating the degree of contamination. 
These extra parameters enable the cBB-D to capture empirical skewness and excess kurtosis better, and, by extension, address the excess of extreme observations compared to the reference BB-D. 
Furthermore, a cBB regression model (cBB-RM) is developed by incorporating the cBB-D into a regression framework, allowing for the inclusion of available and relevant covariates. 
Unlike traditional regression models, we do not limit the regression to the parameter $\pi$ only, but, considering convenient link functions, we extend the regression to all the parameters of the cBB-D and, to further increase the flexibility of the method, we also allow for different covariates on each parameter.

The paper is set out as follows: 
The proposed cBB-D and cBB-RM are presented in Section \ref{Section Methodology}. 
An expectation-maximization (EM) algorithm for maximum likelihood estimation is outlined in Section \ref{Section maximum likelihood estimation}, along with a discussion on strategies for initializing the parameters and convergence criteria considered. A simulation study in Section \ref{Section simulation} illustrates its parameter recovery ability, followed by a sensitivity analysis to investigate the impact of extreme observation on the estimations. In Section  \ref{Section application}, we focus on two distinct environmental datasets, namely the mule deer survival dataset and the Eurobarometer survey data, to illustrate the viability of the cBB-D and cBB-RM as alternative models for overdispersed bounded count data. Finally, conclusions are drawn in Section \ref{Section conclusion}.

\section{Methodological proposals}
\label{Section Methodology}

In this section, we introduce the cBB-D (Section~\ref{Section contaminated beta binomial distribution}) and the cBB-RM (Section~\ref{Section Contaminated beta binomial regression model}).

\subsection{The contaminated beta-binomial distribution} \label{Section contaminated beta binomial distribution}

The PMF of the proposed cBB-D is
\begin{align} \label{pdf contaminated beta binomial}
    f_{\text{cBB}_m}(y;\pi,\sigma,\delta,\eta) = (1-\delta)\underbrace{f_{\text{BB}_m}(y;\pi,\sigma)}_{\text{reference }}+\delta\underbrace{ f_{\text{BB}_m}(y;\pi,\eta\sigma)}_{\text{contaminant}},
\end{align}
where $\pi\in(0,1)$, $\sigma>0$, $\delta\in(0,1)$, and $\eta > 1$. 
If $Y$ has the PMF given in \eqref{pdf contaminated beta binomial}, we will simply write $Y\sim c\mathcal{BB}_m(\pi,\sigma,\delta,\eta)$. 
Due to the bimodal nature of the BB-D -- which can be bell-shaped, U-shaped, J-shaped, or reverse-J-shaped depending on its parameter values -- the cBB-D may also be bimodal, or even trimodal. This flexibility allows the cBB-D to model W-shaped data, where data is clustered at both tails like the U-shape of the BB-D, while retaining an additional central mode (\citealp{gallop2013model}, \citealp{keller2022w}, and \citealp{korkmaz2020new}). 
The bimodal nature of the BB-D is illustrated in \figurename~\ref{plot BB} for varying values of $\sigma$, while examples of the cBB-D are illustrated in \figurename~\ref{plot cBB} for different choices of $\delta$ and $\eta$.
\begin{figure}[h!]
	\centering
	\begin{subfigure}[h]{0.48\textwidth}
		\centering
  \includegraphics[scale=0.48]{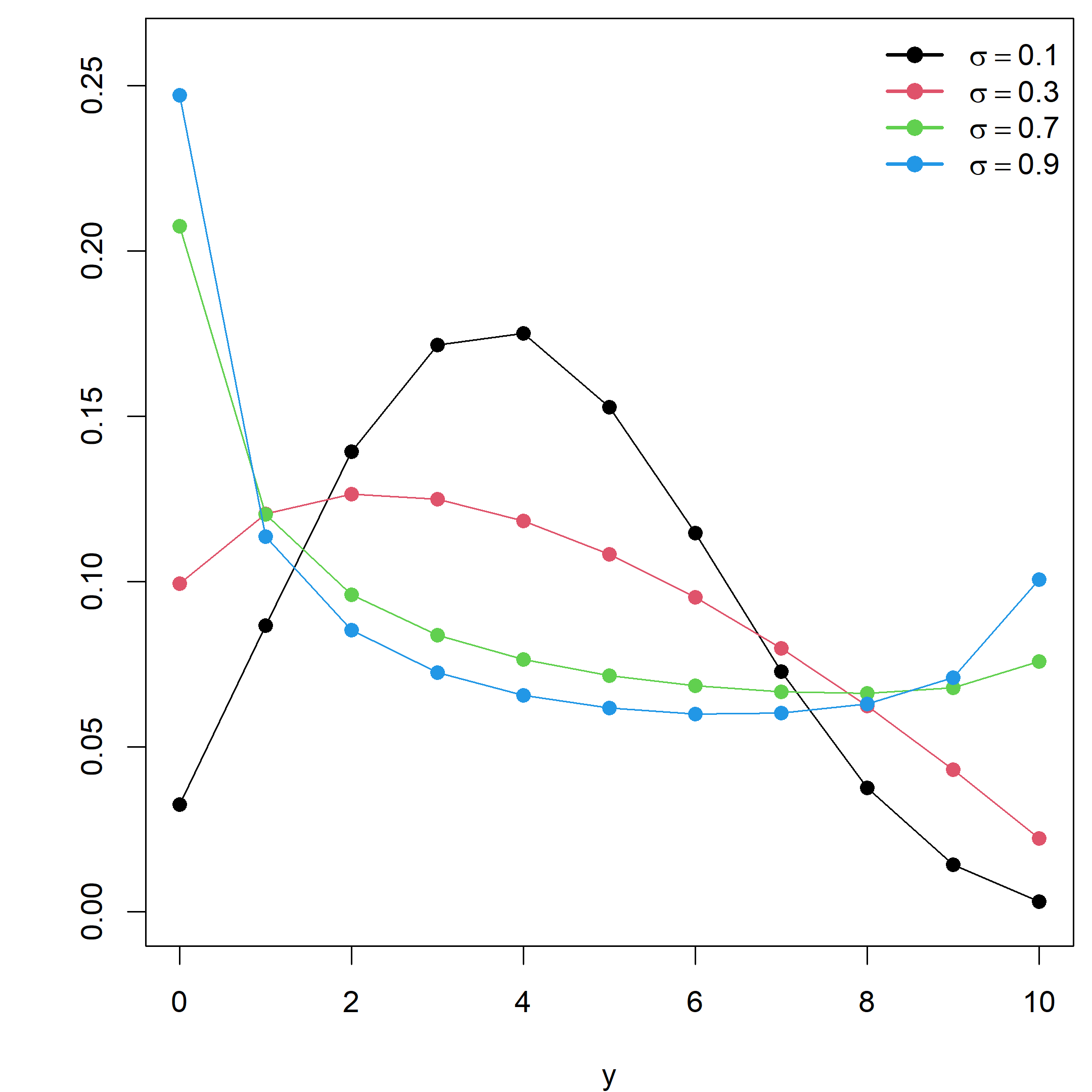}
		\caption{$\pi=0.4$  }
	\end{subfigure}
 \begin{subfigure}[h]{0.48\textwidth}
		\centering
		\includegraphics[scale=0.48]{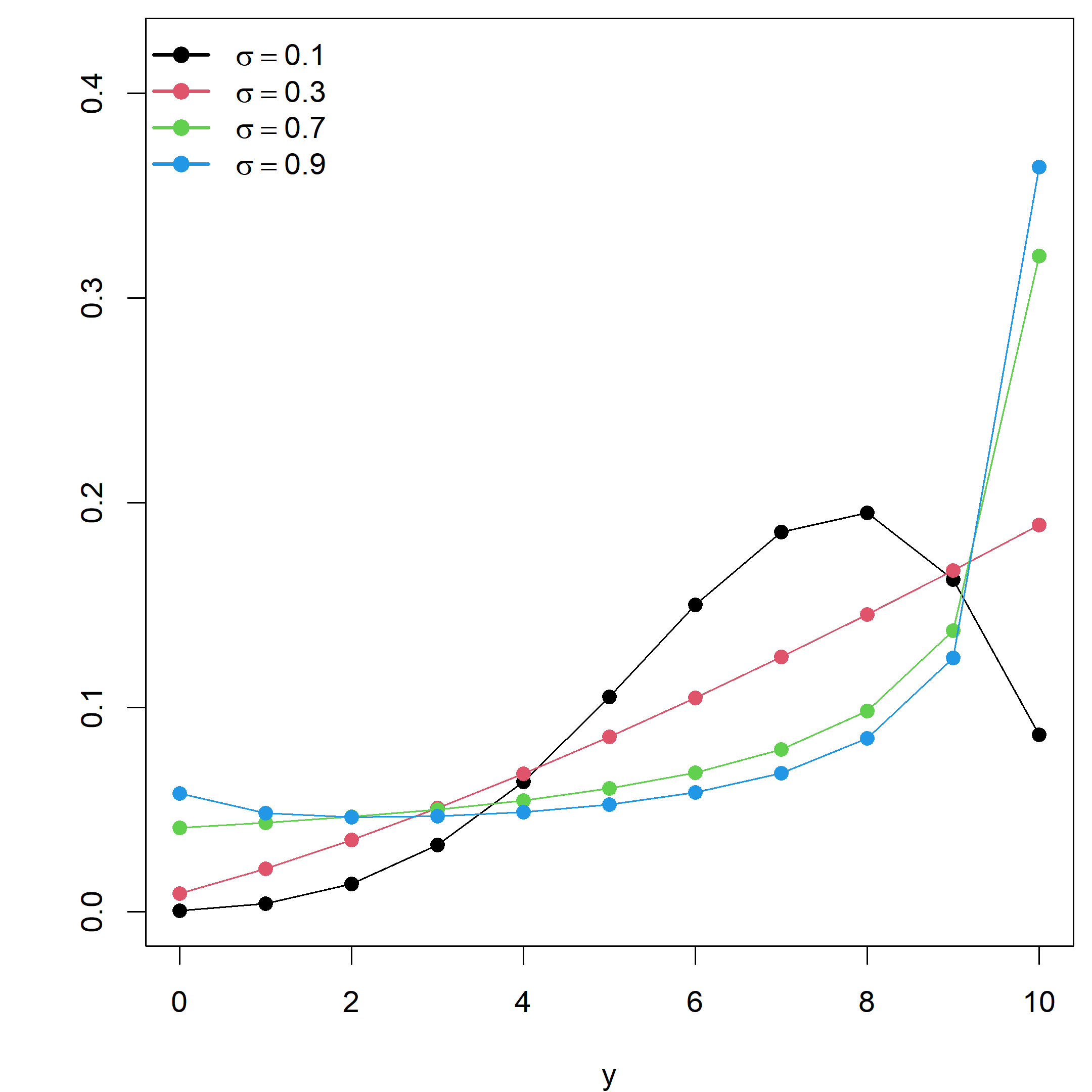}
		\caption{$\pi=0.7$ }
	\end{subfigure}
	\caption{Plots of the BB-D \eqref{pmf mean beta binomial distribution} for different values of $\sigma$ when $m=10$.}
\label{plot BB}
\end{figure}

\begin{figure}[h!]
	\centering
	\begin{subfigure}[h]{0.48\textwidth}
		\centering
  \includegraphics[scale=0.48]{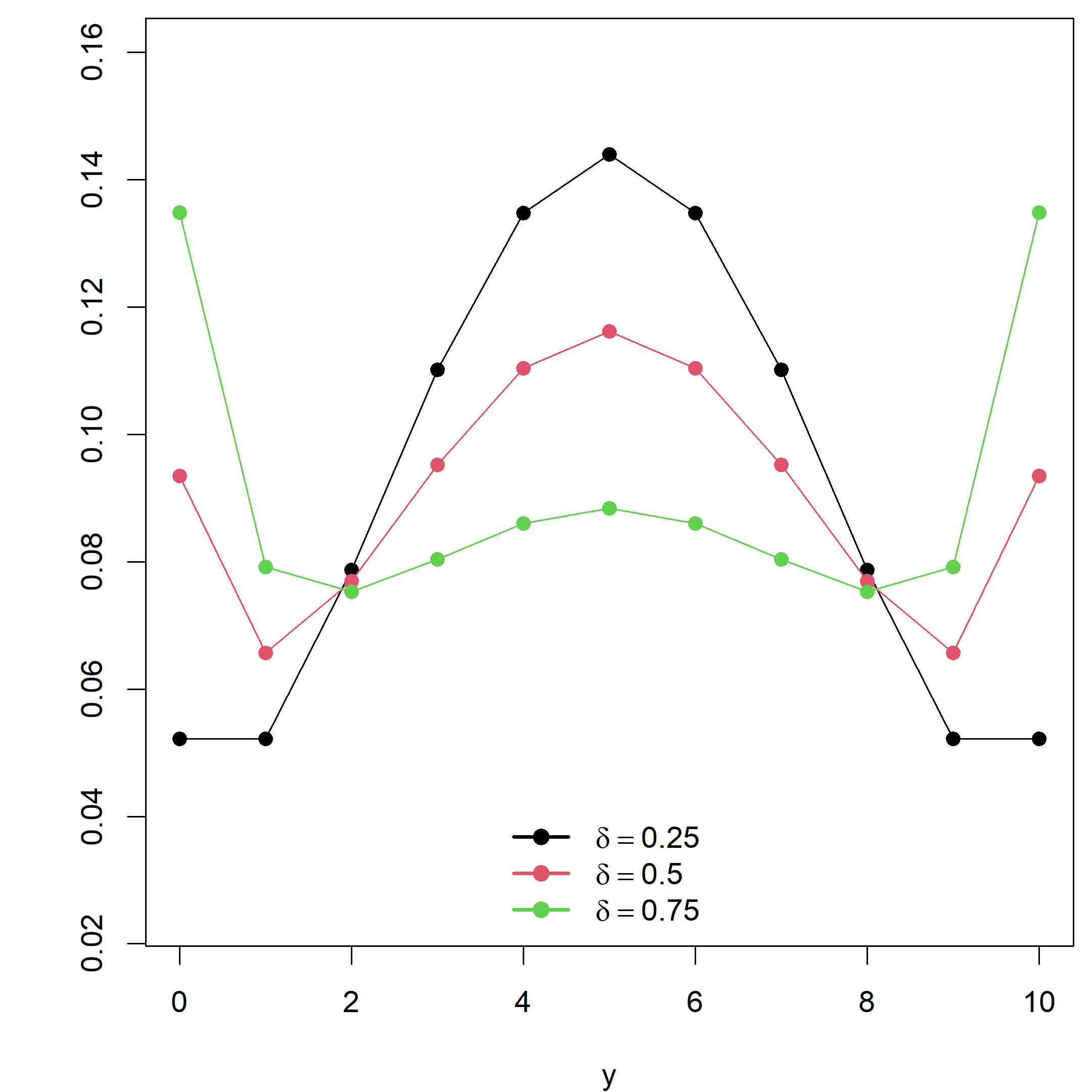}
		\caption{Different values of $\delta$.}
	\end{subfigure}
 \begin{subfigure}[h]{0.48\textwidth}
		\centering
		\includegraphics[scale=0.48]{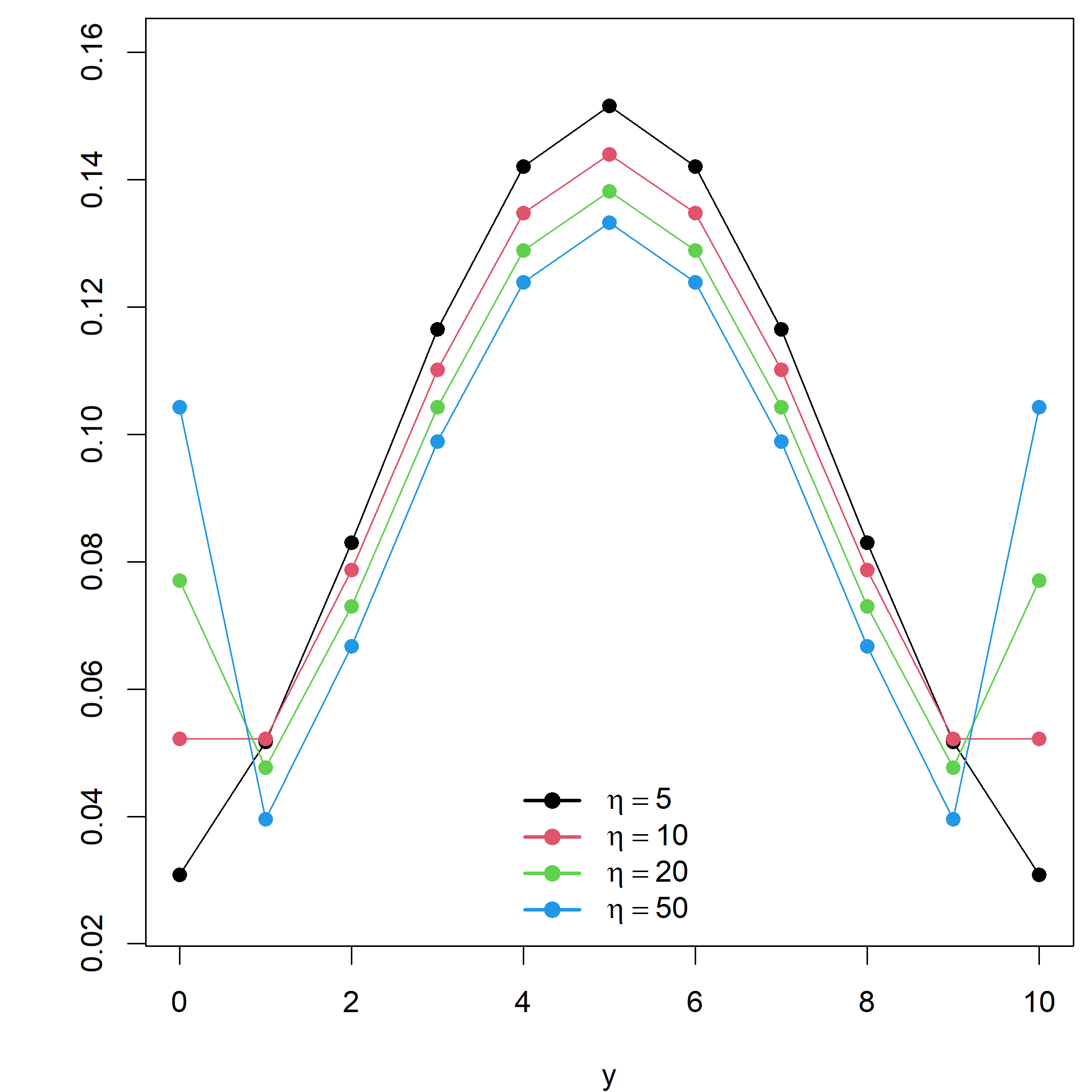}
		\caption{Different values of $\eta$. }
	\end{subfigure}
	\caption{Plots of the cBB-D \eqref{pdf contaminated beta binomial} with $m=10$, $\pi =0.5$, and $ \sigma=0.1$ for different values of $\delta$ (when $\eta=10$) and $\eta$ (when $\delta=0.25$).}
\label{plot cBB}
\end{figure}

The moments, or shape characteristics, of practical interest of $Y\sim c\mathcal{BB}_m(\pi,\sigma,\delta,\eta)$ are derived in \ref{ProofinApp} and are:
\begin{equation}
    \mathrm{E}_{\text{cBB}_m}(Y;\pi)=m\pi, \label{ec:cBB mean}
\end{equation}
\begin{align}\label{eq cBB var}
        \text{Var}_{\text{cBB}_m}(Y;\pi,\sigma,\delta,\eta) 
        & = (1-\delta)\text{Var}_{\text{BB}_m}(Y;\pi,\sigma)+\delta\text{Var}_{\text{BB}_m}(Y;\pi,\eta\sigma)\nonumber\\
        &= \frac{m\pi(1-\pi)\left[(1-\delta)(1+m\sigma)(1+\eta\sigma)+\delta(1+m\eta\sigma)(1+\sigma)\right]}{(1+\sigma)(1+\eta\sigma)},
    \end{align}
    \begin{align}\label{eq cBB skew}
    \text{Skew}_{\text{cBB}_m}(Y;\pi,\sigma,\delta,\eta)=&\frac{1-\delta}{\left(\text{Var}_{\text{cBB}_m}(Y;\pi,\sigma,\delta,\eta)\right)^\frac{3}{2}}\left(\frac{m\pi(1-\pi)(1-2\pi)(1+m\sigma )(1+2m\sigma )}{(1+\sigma )(1+2\sigma )}\right)\notag\\
          &+
          \frac{\delta}{\left(\text{Var}_{\text{cBB}_m}(Y;\pi,\sigma,\delta,\eta)\right)^\frac{3}{2}}\left(\frac{m\pi(1-\pi)(1-2\pi)(1+m\sigma \eta)(1+2m\sigma \eta)}{(1+\sigma \eta)(1+2\sigma \eta)}\right),
          \end{align}
 \begin{align}\label{eq cBB kurt}
    &\text{ExKurt}_{\text{cBB}_m}(Y;\pi,\sigma,\delta,\eta)\notag\\
    &=-3+\frac{m \pi   (1-\pi ) }{\left[{\text{Var}_{\text{cBB}}(Y;\pi,\sigma,\delta,\eta)}\right]^2}\notag\\
    &\times\left(\frac{(1-\delta)(m \sigma +1) \left(6 (3 (\pi -1) \pi +1) m^2 \sigma ^2+3 m \sigma  (2-(\pi -1) \pi  (m-6))-3 (\pi -1) \pi  (m-2)-\sigma +1\right)}{(\sigma +1) (2 \sigma +1) (3 \sigma +1)}\right.\notag\\
    &+\left.\frac{\delta(\eta  m \sigma +1) \left(-\eta  \sigma +6 \eta ^2 (3 (\pi -1) \pi +1) m^2 \sigma ^2+3 \eta  m \sigma  (2-(\pi -1) \pi  (m-6))-3 (\pi -1) \pi  (m-2)+1\right)}{(\eta  \sigma +1) (2 \eta  \sigma +1) (3 \eta  \sigma +1)}\right).
\end{align}
The variance ranges between $\left(0,\frac{m\left[(1-\delta)(1+m\sigma)(1+\eta\sigma)+\delta(1+m\eta\sigma)(1+\sigma)\right]}{4(1+\sigma)(1+\eta\sigma)}\right]$,  reaching a maximum at $\pi=0.5$ and tends to 0 as $\pi$ approach $0^+$ or $1^-$. The variance increases with $\sigma$ and, since $\delta \in (0,1)$ and $\eta>1$, the variance of the cBB-D is always greater than that of the BB-D. Moreover, the variance increases as $\delta$ and $\eta$ increase, as illustrated in \figurename~\ref{plot variance cBB} for different values of $\delta$ and $\eta$. Similarly, the skewness of the cBB-D is consistently larger in magnitude than the BB-D counterparts, as illustrated in \figurename~\ref{plot skewness cBB}. 
The kurtosis is at a minimum for $\pi=0.5$, with the effect of $\delta$ and $\eta$ illustrated in \figurename~\ref{plot kurt cBB}. We observe that increasing either $\delta$ or $\eta$ leads to a rise in excess kurtosis for some values of $\pi$, while for other values, it causes a decrease - mirroring the effect of increasing $\sigma$ in the BB-D case.
Since the cBB-D moments depend on four parameters, rather than just $\pi$ and $\sigma$ as in the BB-D, the inclusion of $\delta$ and $\eta$ enhances the model's flexibility. 
This added flexibility allows the cBB-D to better accommodate for possible extreme observations and effectively address overdispersion observed in the BB-D. Moreover, the moments of the BB-D emerge as limiting cases of the cBB-D, reinforcing its ability to generalize and extend the BB-D.
 
\begin{figure}[!h]
	\centering
	\begin{subfigure}[h]{0.48\textwidth}
		\centering
  \includegraphics[scale=0.48]{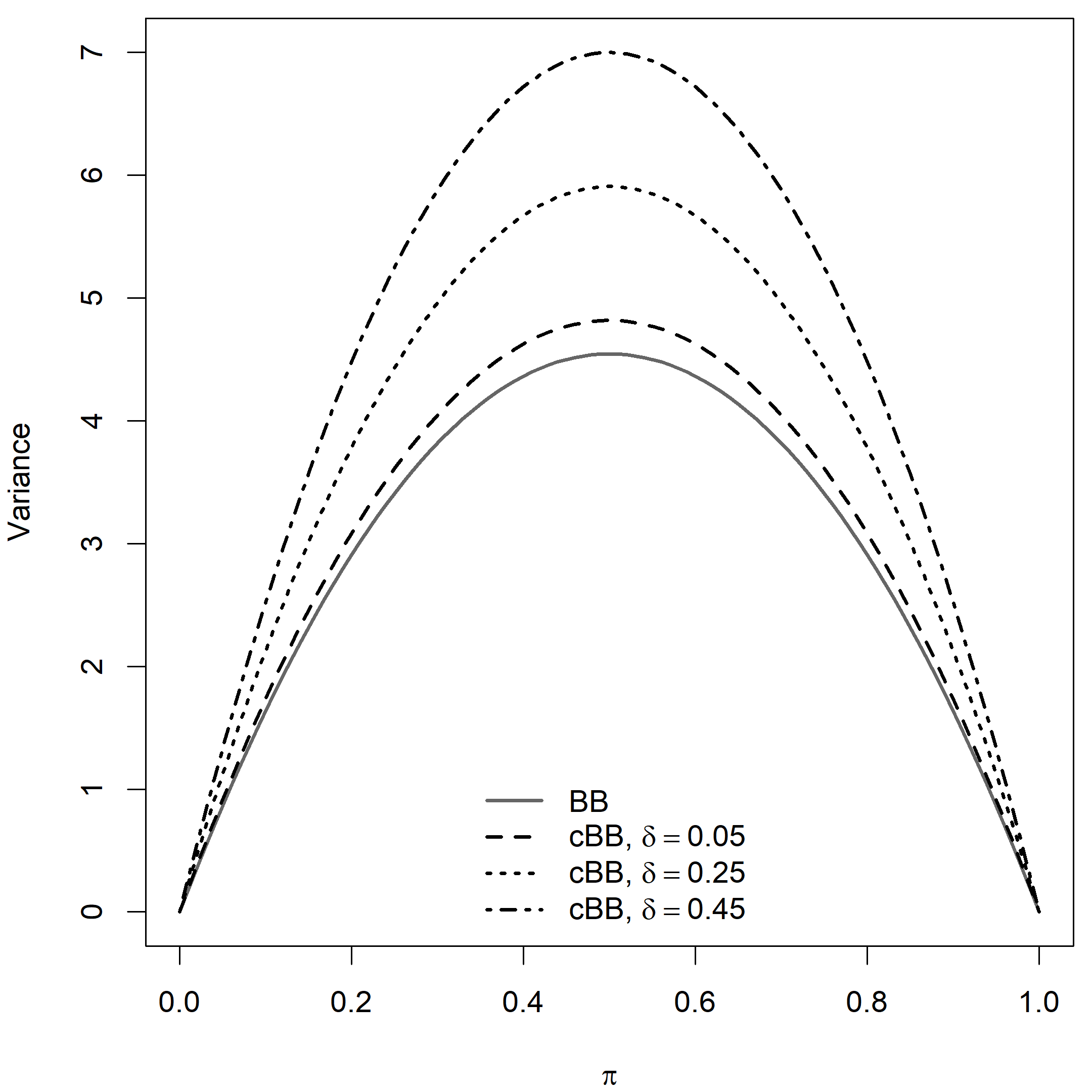}
		\caption{Different values of $\delta$.}
	\end{subfigure}
 \begin{subfigure}[h]{0.48\textwidth}
		\centering
		\includegraphics[scale=0.48]{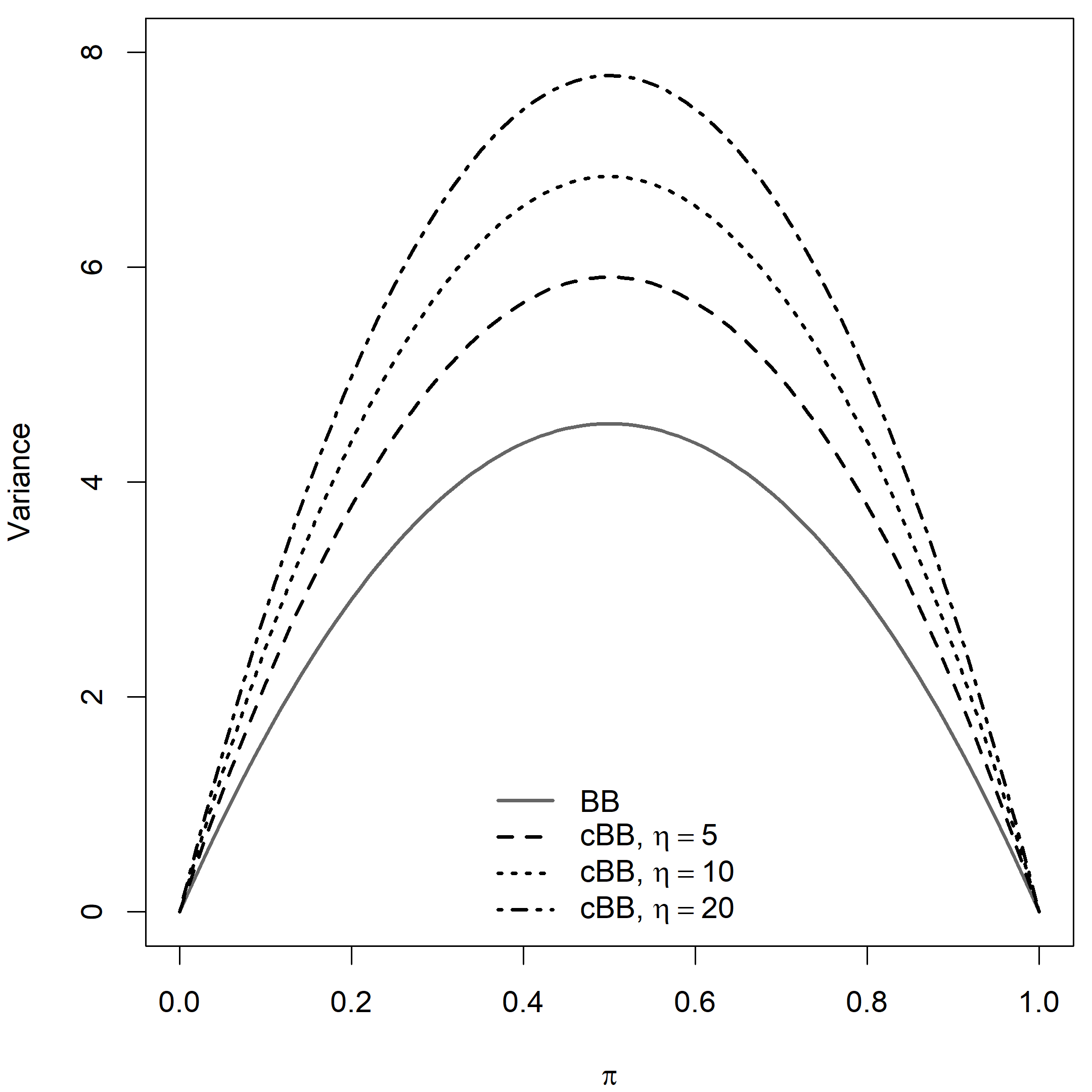}
		\caption{Different values of $\eta$. }
	\end{subfigure}
 	\begin{subfigure}[h]{0.48\textwidth}
		\centering
  \includegraphics[scale=0.48]{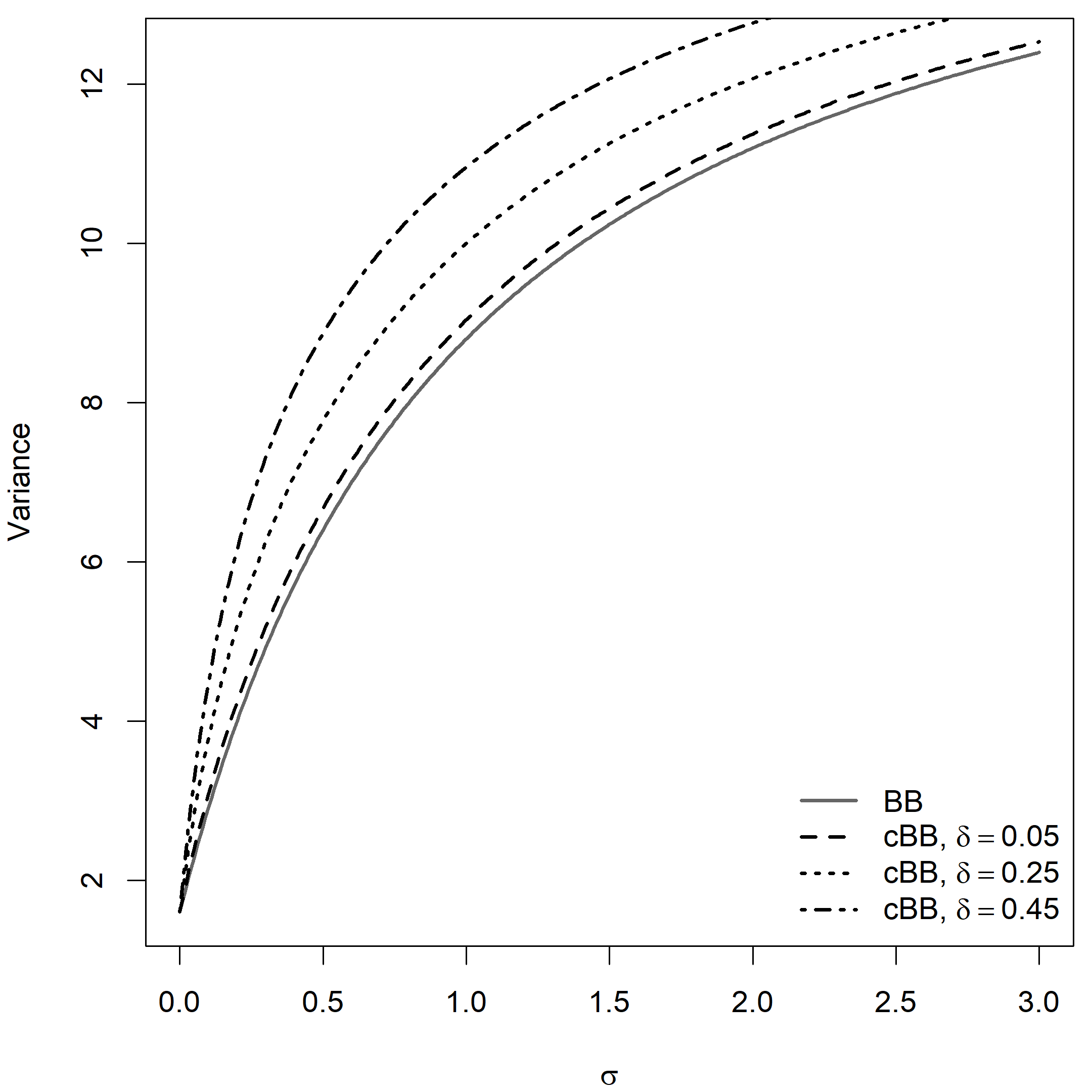}
		\caption{Different values of $\delta$.}
	\end{subfigure}
 \begin{subfigure}[h]{0.48\textwidth}
		\centering
		\includegraphics[scale=0.48]{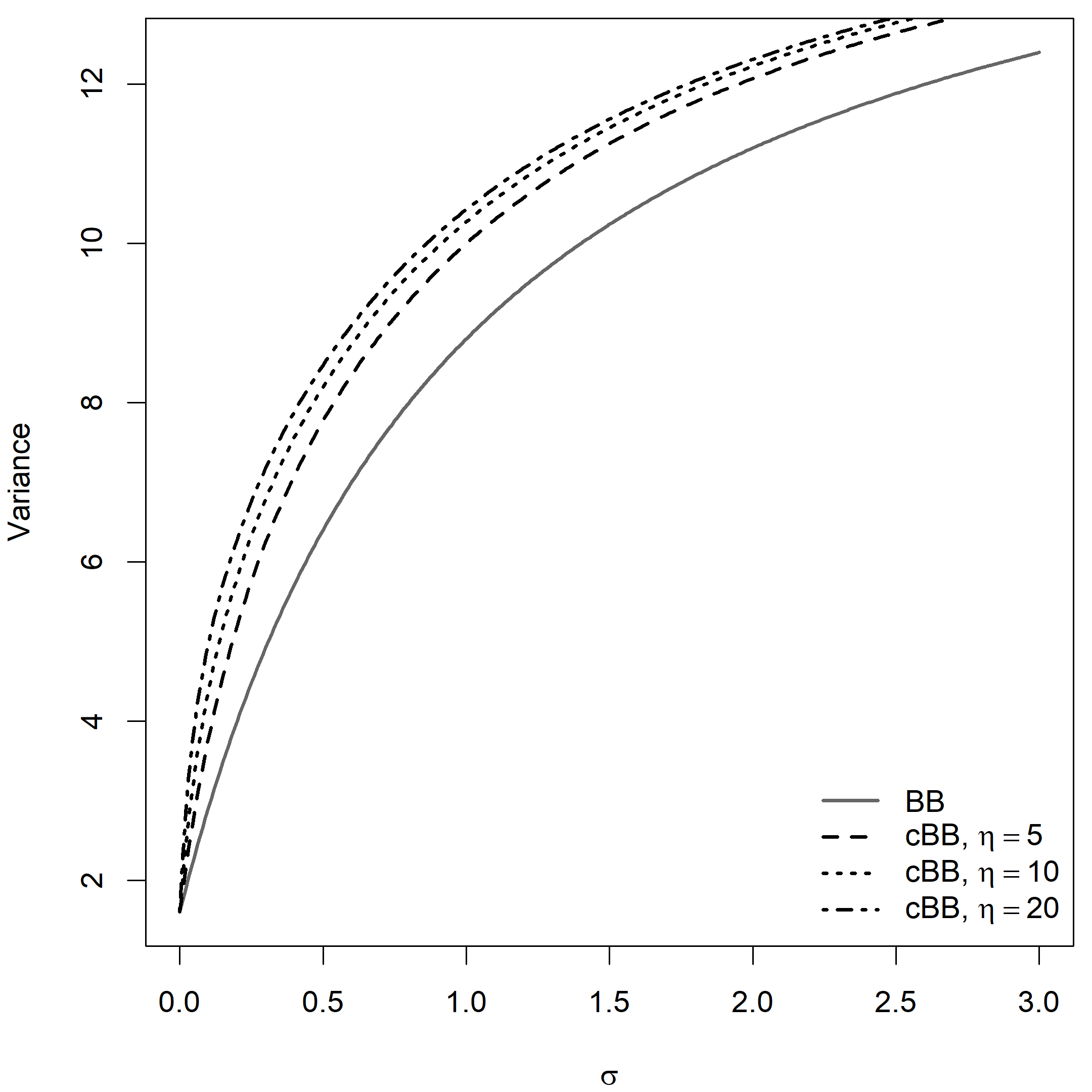}
		\caption{Different values of $\eta$. }
	\end{subfigure}
	\caption{Examples illustrating higher variance \eqref{eq cBB var} of the cBB-D with $m=10$, $\pi =0.2$, and $\sigma=0.1$ compared to the BB-D variance \eqref{eq variance betabinomial} for $\pi \in (0,1)$ and increasing values of $\sigma$ for varying values of $\delta$ (when $\eta=5$) and $\eta$ (when $\delta=0.25$).}
\label{plot variance cBB}
\end{figure}

\begin{figure}[h!]
	\centering
	\begin{subfigure}[h]{0.48\textwidth}
		\centering
  \includegraphics[scale=0.48]{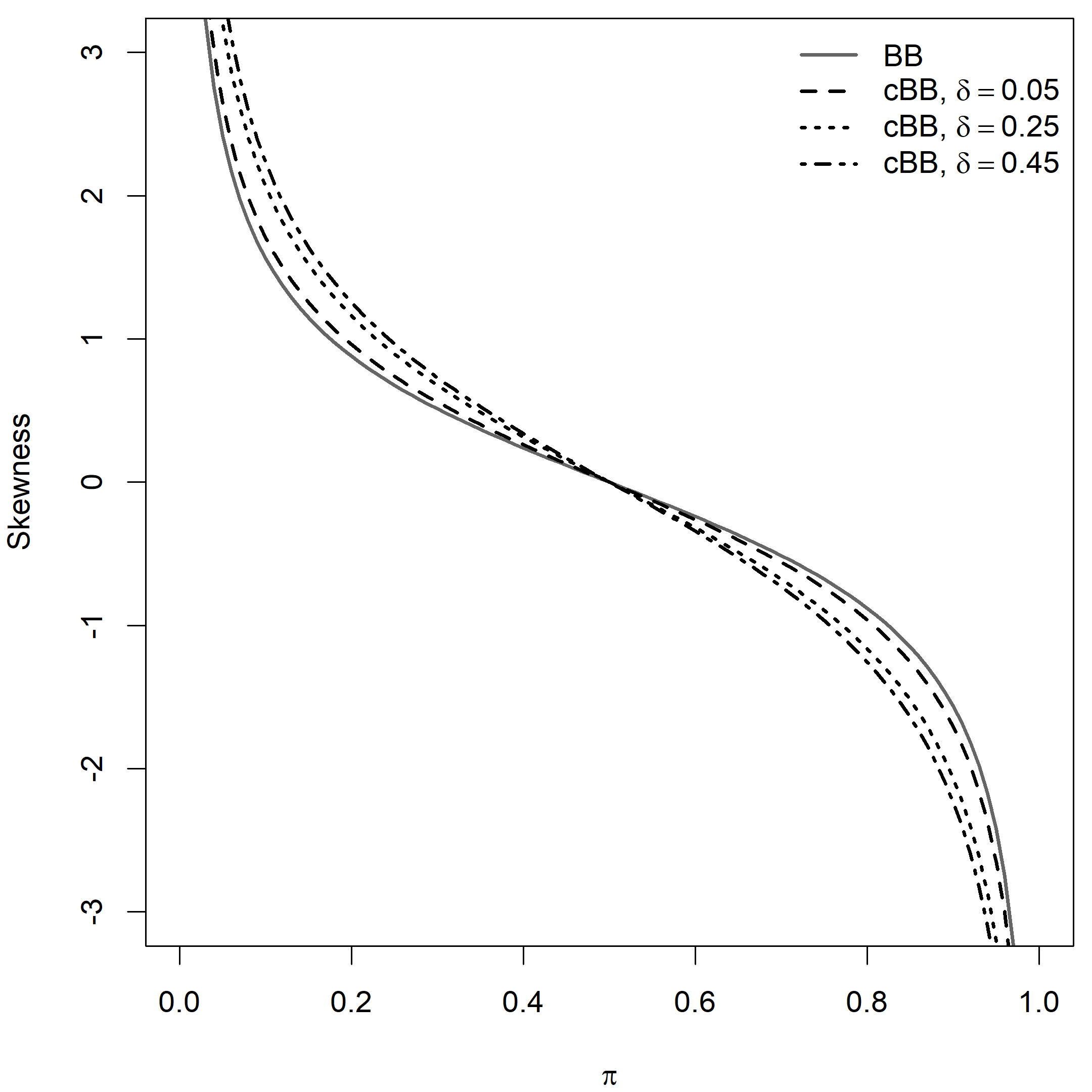}
		\caption{Different values of $\delta$.}
	\end{subfigure}
 \begin{subfigure}[h]{0.48\textwidth}
		\centering
		\includegraphics[scale=0.48]{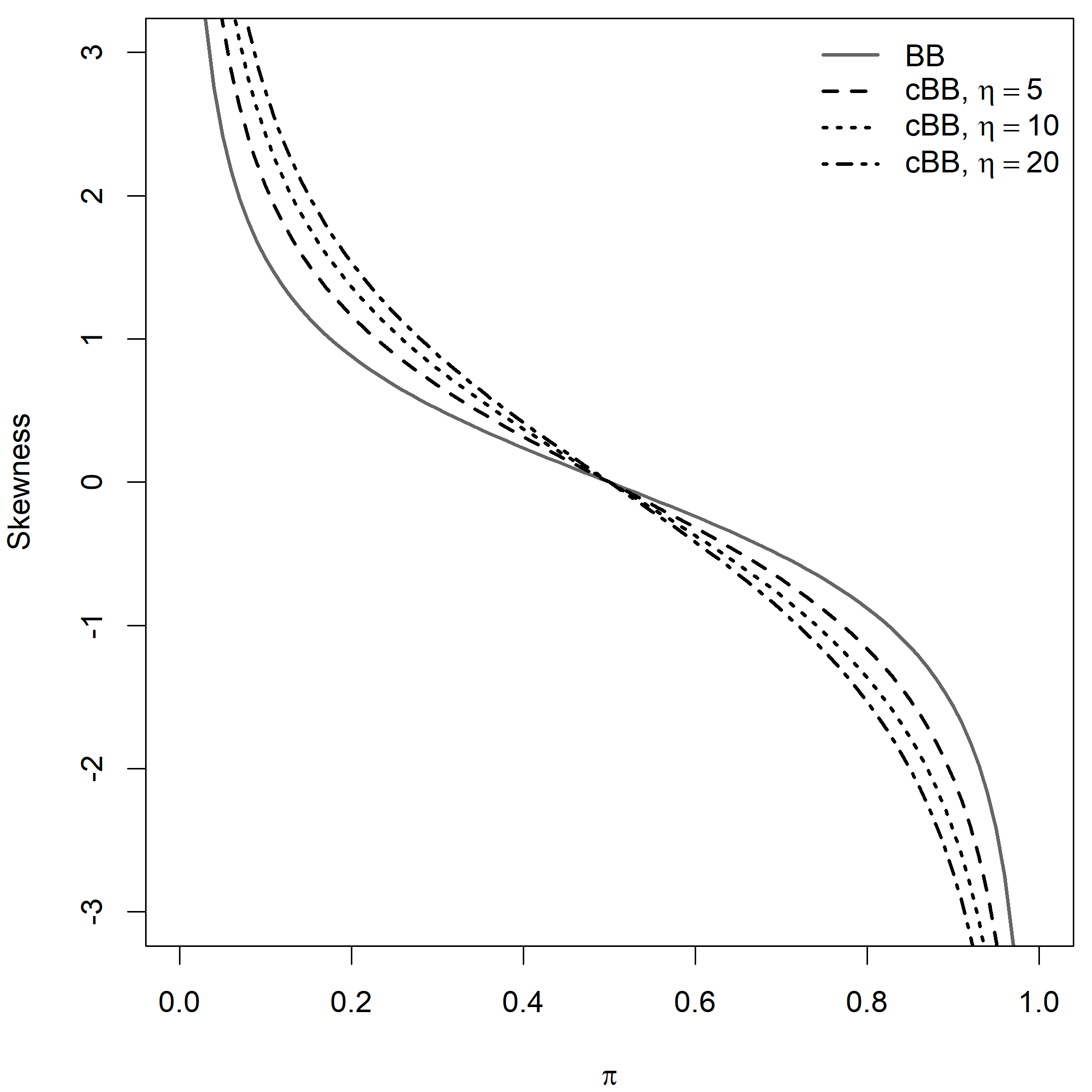}
		\caption{Different values of $\eta$. }
	\end{subfigure}
 	\begin{subfigure}[h]{0.48\textwidth}
		\centering
  \includegraphics[scale=0.48]{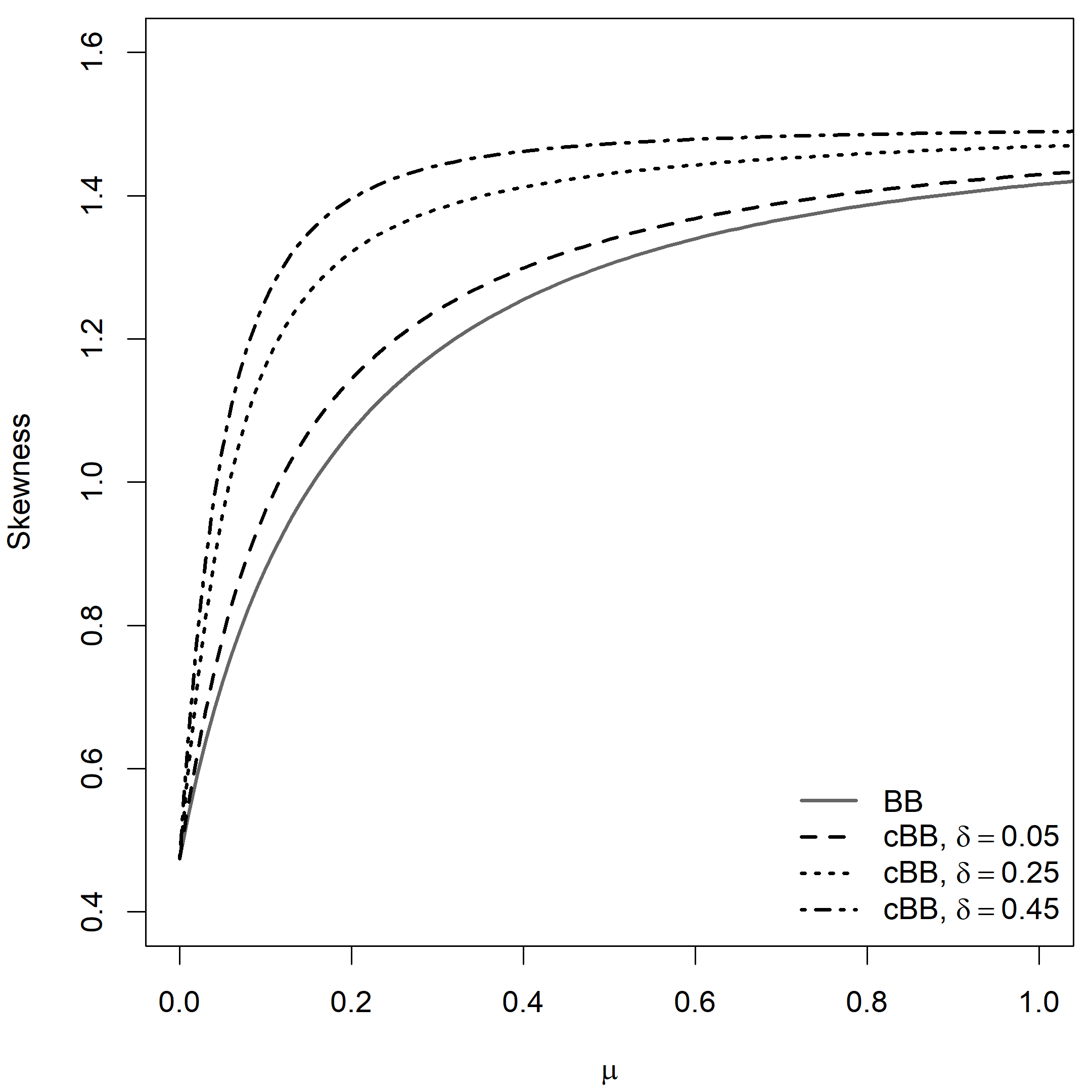}
		\caption{Different values of $\delta$.}
	\end{subfigure}
 \begin{subfigure}[h]{0.48\textwidth}
		\centering
		\includegraphics[scale=0.48]{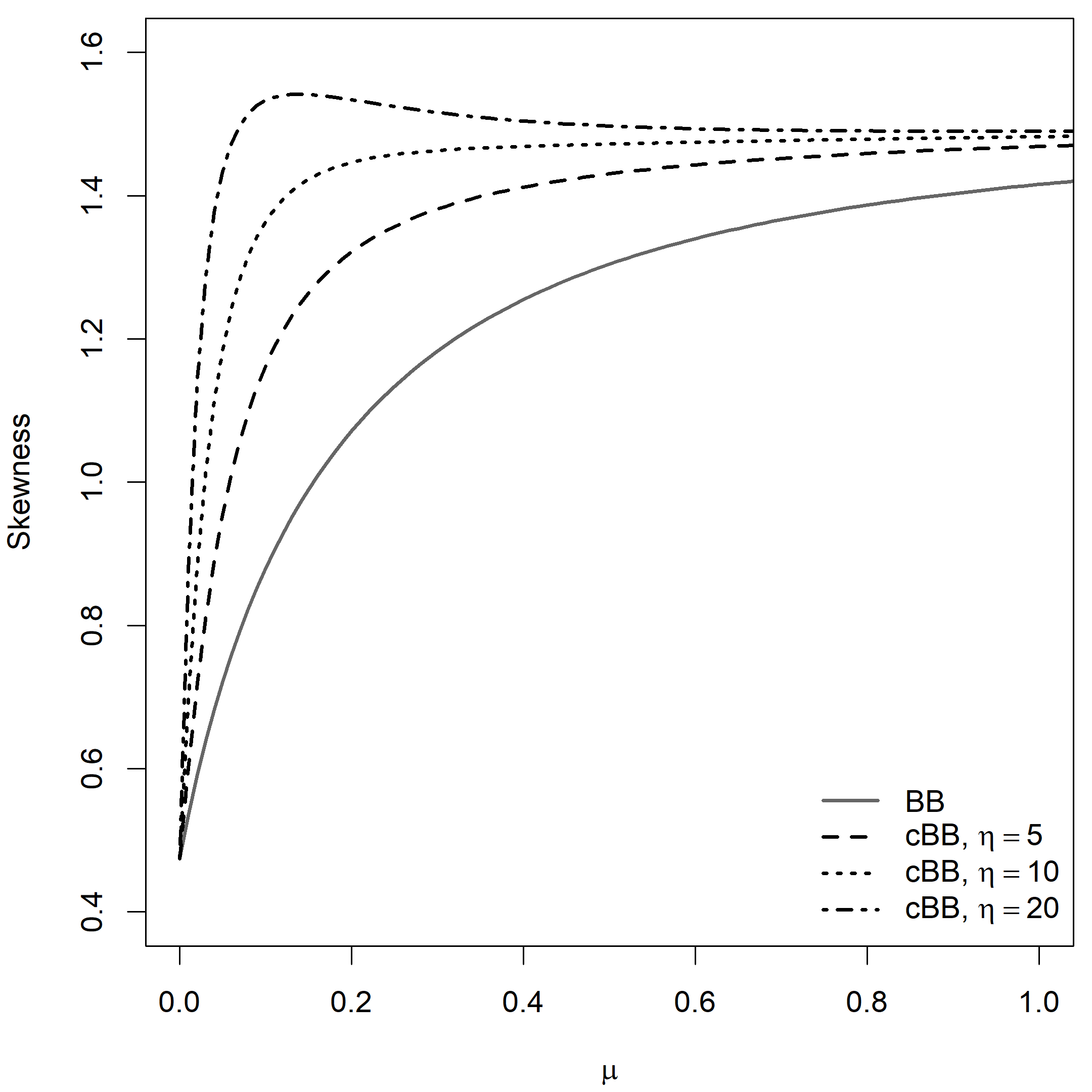}
		\caption{Different values of $\eta$. }
	\end{subfigure}
	\caption{Examples illustrating flexible skewness of the cBB-D \eqref{eq cBB skew} with $m=10$, $\pi =0.2$, and $\sigma=0.1$ compared to the BB-D skewness \eqref{eq BB skew} for $\pi \in (0,1)$ and increasing values of $\sigma$ for varying values of $\delta$ (when $\eta=5$) and $\eta$ (when $\delta=0.25$).}
\label{plot skewness cBB}
\end{figure}

\begin{figure}[h!]
	\centering
	\begin{subfigure}[h]{0.48\textwidth}
		\centering
  \includegraphics[scale=0.48]{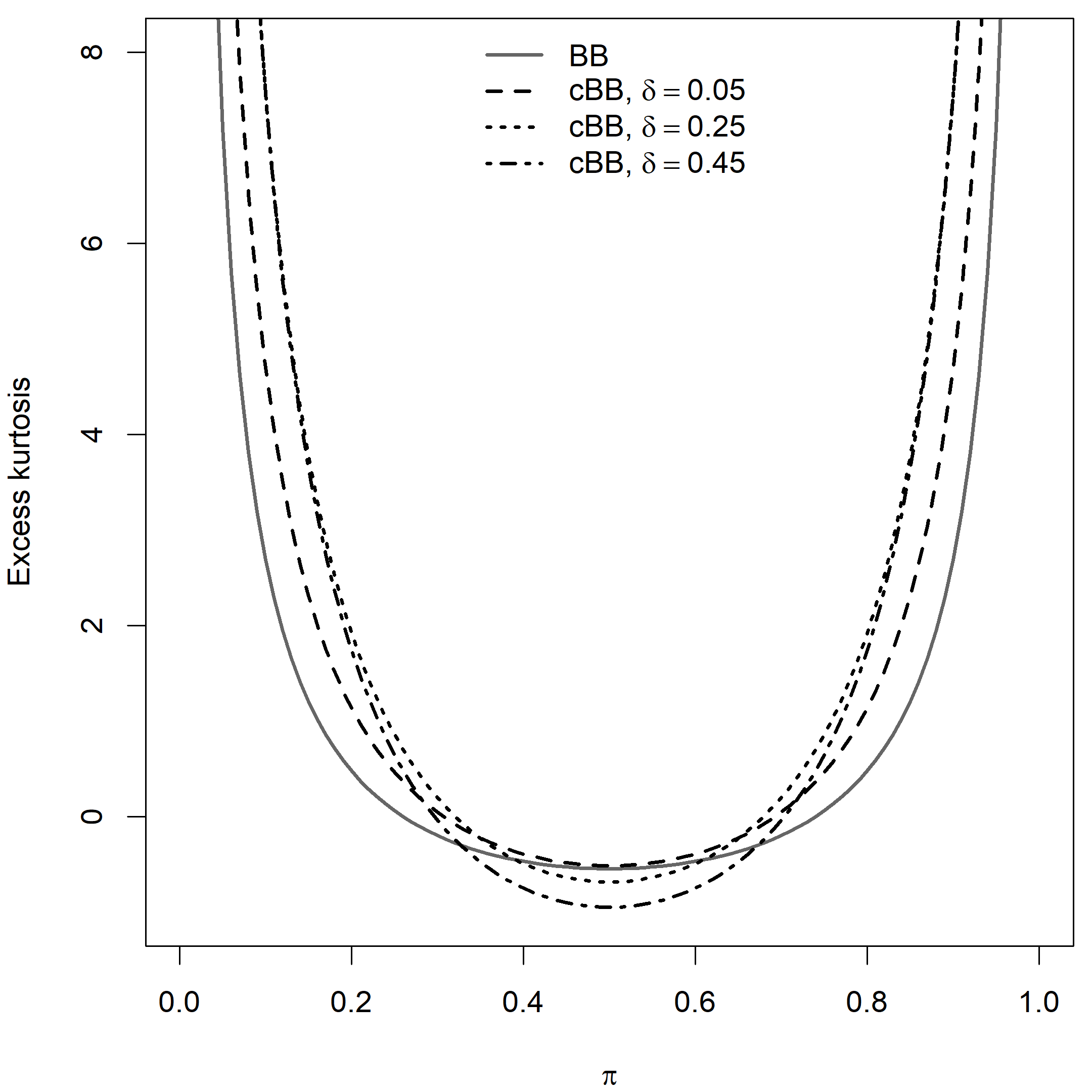}
		\caption{Different values of $\delta$.}
	\end{subfigure}
 \begin{subfigure}[h]{0.48\textwidth}
		\centering
		\includegraphics[scale=0.48]{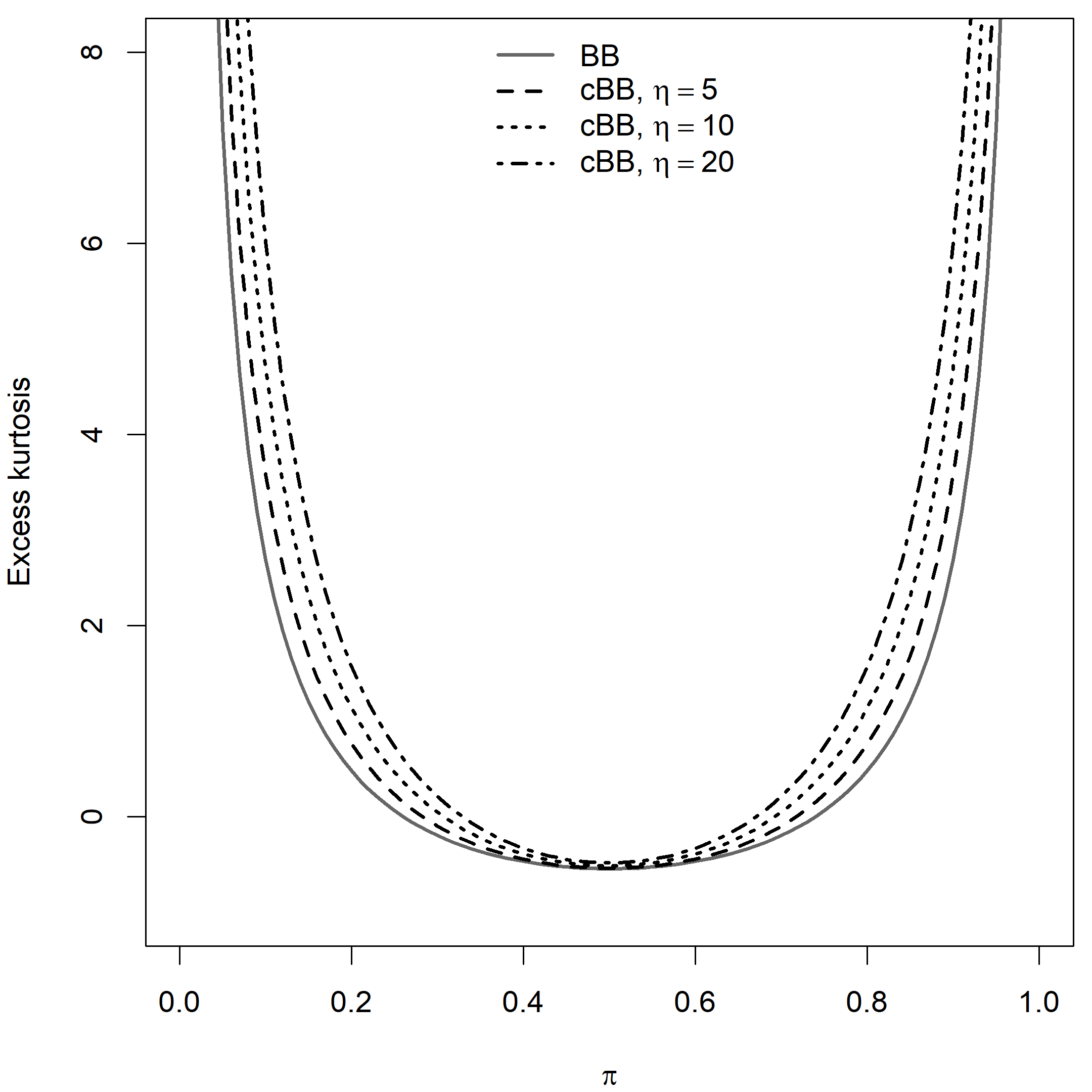}
		\caption{Different values of $\eta$. }
	\end{subfigure}
 	\begin{subfigure}[h]{0.48\textwidth}
		\centering
  \includegraphics[scale=0.48]{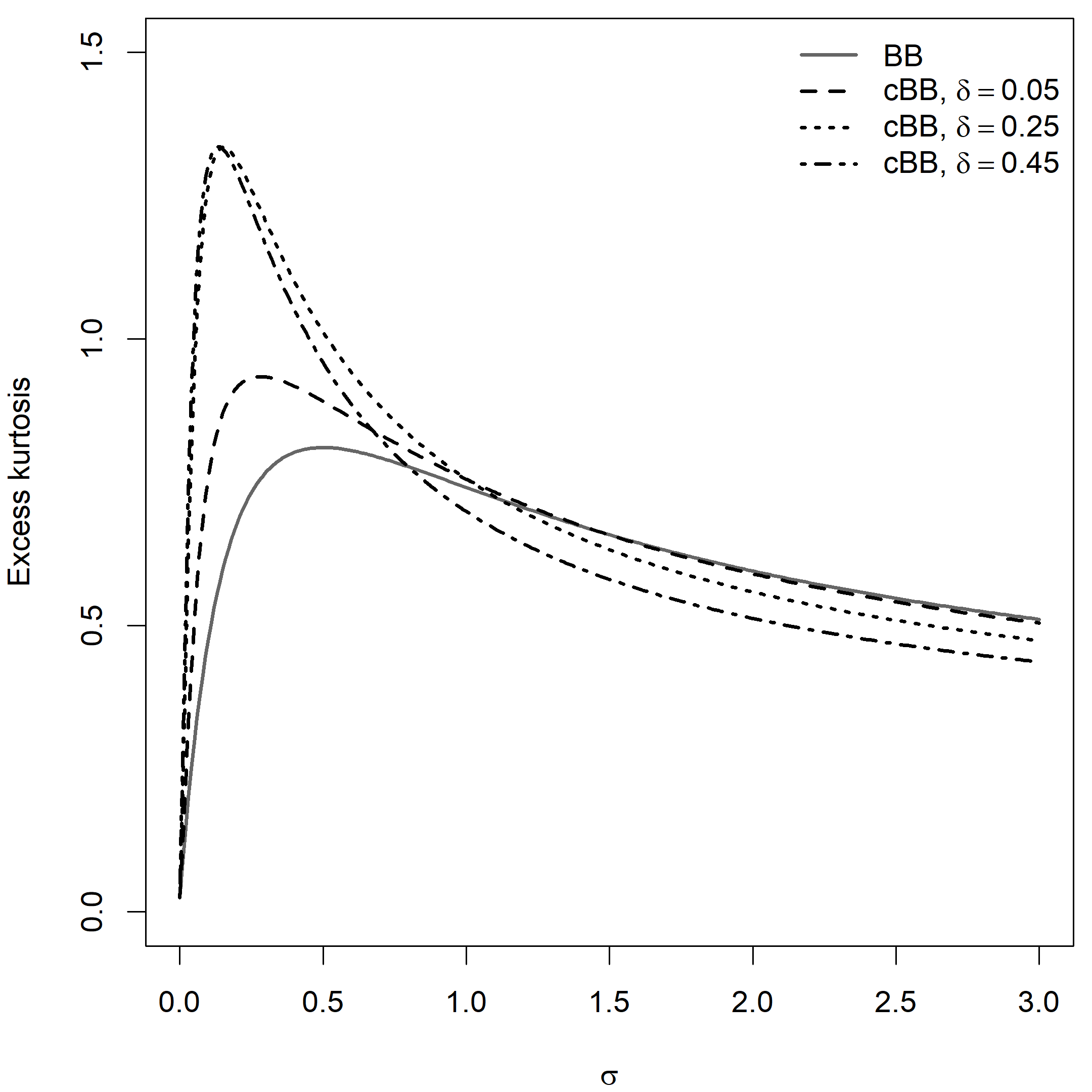}
		\caption{Different values of $\delta$.}
	\end{subfigure}
 \begin{subfigure}[h]{0.48\textwidth}
		\centering
		\includegraphics[scale=0.48]{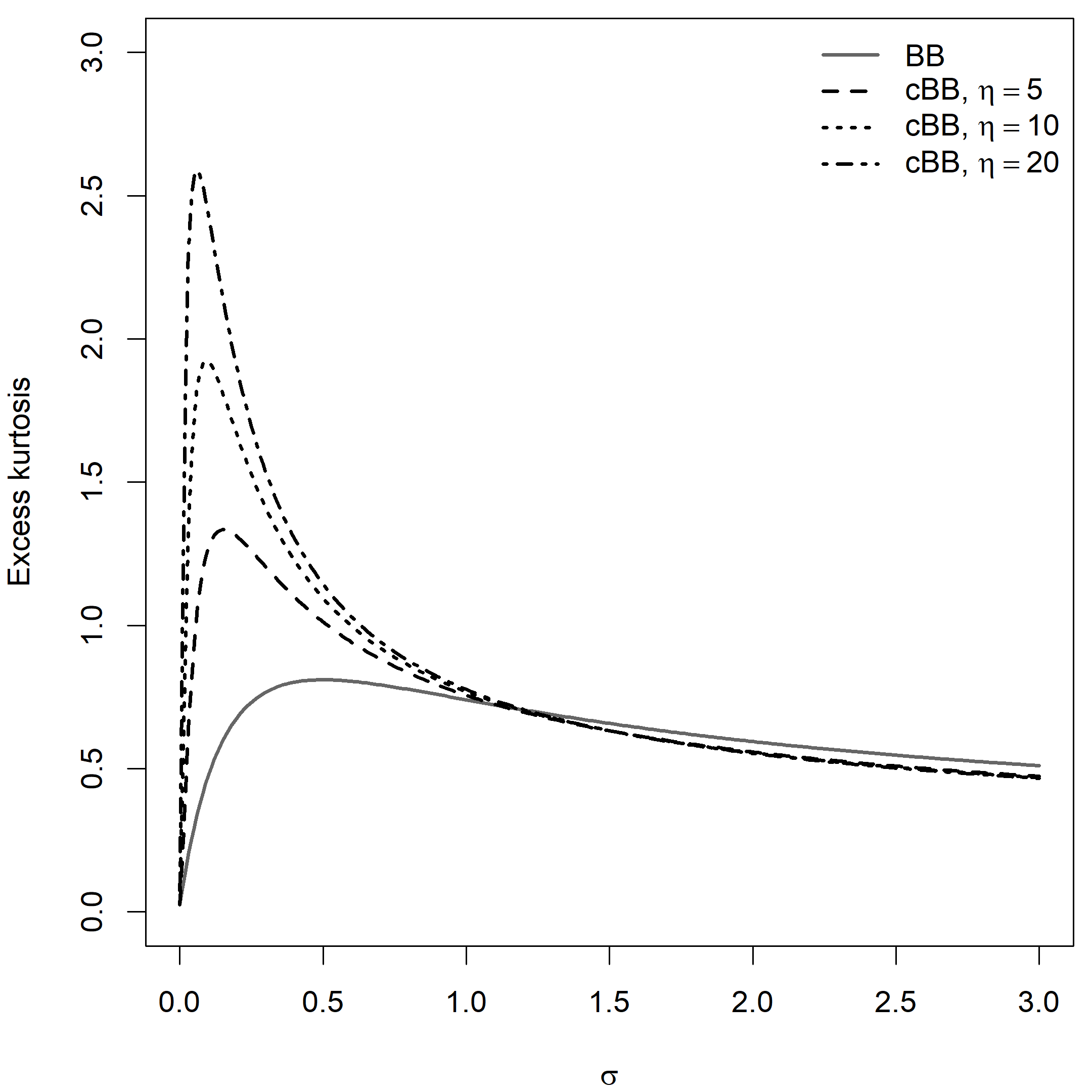}
		\caption{Different values of $\eta$. }
	\end{subfigure}
	\caption{Examples illustrating larger excess kurtosis of the cBB-D \eqref{eq cBB kurt} with $m=10$, $\pi =0.2$, and $\sigma=0.1$ compared to the BB-D excess kurtosis \eqref{eq BB kurt} for $\pi \in (0,1)$ and increasing values of $\sigma$ for varying values of $\delta$ (when $\eta=5$) and $\eta$ (when $\delta=0.25$).}
\label{plot kurt cBB}
\end{figure}

In Proposition \ref{proposition 1} we explore some limiting cases of the cBB-D.
\begin{prop}
\label{proposition 1}
    Let $Y\sim c\mathcal{BB}_m(\pi,\sigma,\delta,\eta)$, then:
\begin{enumerate}[label=(\itshape\alph*\upshape)]
	\item\label{item:prop1} if  $\delta \to 1^-$, then $Y\overset{D}{\to}\mathcal{BB}_m(\pi,\sigma)$; 
	\item\label{item:prop2} if  $\eta \to 1^+$, then $Y\overset{D}{\to}\mathcal{BB}_m(\pi,\sigma)$;
 \item\label{item:prop3} if  $\delta \to 1^-$ and $\sigma\to0^+$, then $Y\overset{D}{\to}\mathcal{B}_m(\pi)$; 
	\item\label{item:prop4} if  $\eta\to 1^+$ and $\sigma\to0^+$, then $Y\overset{D}{\to}\mathcal{B}_m(\pi)$;
 \end{enumerate}
\end{prop}
\noindent where $\overset{D}{\to}$ denotes convergence in distribution.
\begin{proof}
    See \ref{ProofinApp}.
\end{proof}

\newpage
\FloatBarrier

\subsection{The contaminated beta-binomial regression model}
\label{Section Contaminated beta binomial regression model}

In traditional regression models, the focus is typically placed on modelling only the mean, with the assumption that other parameters, like the dispersion or contamination parameters in our case, remain fixed. 
This can be limiting, especially in situations where, for example, the dispersion in the data may be influenced by covariates and make the assumption of a constant $\sigma$ inappropriate.
Moreover, we allow for different covariates for each parameter: this offers greater flexibility, which enables the model to account for the distinct factors that affect various aspects of the cBB distribution: the mean, dispersion, proportion of contamination, and degree of contamination.

Let $\bx$, $\bu$, $\bv$, and $\bz$ represent possible values of the covariates $\bX$, $\bU$, $\bV$, and $\bZ$, which have dimensions $p$, $q$, $r$, and $s$, respectively. 
These covariates are used to model the parameters $\pi$, $\sigma$, $\delta$, and $\eta$, respectively, of the cBB-D.
The cBB regression model (cBB-RM) is then specified through the following link functions:
\begin{align*}
g_1({\pi}(\bx;\bbeta)) 
& = \logit({\pi}(\bx;\bbeta))
= \widetilde{\bx}'\bbeta,\\
g_2(\sigma(\bu;\balpha)) 
& = \log(\sigma(\bu;\balpha)) 
= \widetilde{\bu}'\balpha,\\
g_3({\delta}(\bv;\bgamma)) 
& =  \logit(\delta(\bv;\bgamma))
=  \widetilde{\bv}'\bgamma,\\
g_4(\eta(\bz;\blambda)) 
& =  \log(\eta(\bz;\blambda) - 1) 
=  \widetilde{\bz}'\blambda,
\end{align*}
where $\bbeta=(\beta_0,\beta_1,\dots,\beta_p)'$, $\balpha=(\alpha_0,\alpha_1,\dots,\alpha_q)'$, $\bgamma=(\gamma_0,\gamma_1,\dots,\gamma_r)'$, and $\blambda=(\lambda_0,\lambda_1,\dots,\lambda_s)$ are vectors of unknown regression coefficients, $\widetilde{\bx}=(1,\bx')'$, $\widetilde{\bu}=(1,\bu')'$, $\widetilde{\bv}=(1,\bv')'$, and $\widetilde{\bz}=(1,\bz')'$ account for the intercept, and $\mathrm{log}$ represents the natural logarithm.
Naturally, the considered link functions, even if the most commonly used, are only examples of possible functions that can be considered.
The inverse of the considered link functions leads to
\begin{align*}
 \pi(\bx;\bbeta) &= g_1^{-1}(\widetilde{\bx}'\bbeta) =\frac{\mathrm{e}^{\widetilde{\bx}'\bbeta}}{1+\mathrm{e}^{\widetilde{\bx}'\bbeta}},\\
\sigma(\bu;\balpha) &= g_2^{-1}(\widetilde{\bu}'\balpha) =\mathrm{e}^{\widetilde{\bu}'\balpha},\\
\delta(\bv;\bgamma) &= g_3^{-1}(\widetilde{\bv}'\bgamma) =\frac{\mathrm{e}^{\widetilde{\bv}'\bgamma}}{1+\mathrm{e}^{\widetilde{\bv}'\bgamma}},\\
\eta(\bz;\blambda)  &= g_4^{-1}(\widetilde{\bz}'\blambda) =\mathrm{e}^{\widetilde{\bz}'\blambda}+1.
\end{align*}
The conditional distribution of $Y$ according to the cBB-RM can also be written as
\begin{equation}
    Y|\bX=\bx,\bU=\bu,\bV=\bv,\bZ=\bz \sim c\mathcal{BB}_m\left(\pi(\bx;\bbeta),\sigma(\bu;\balpha),\delta(\bv;\bgamma),\eta(\bz;\blambda)\right), 
    \label{eq:cBB-RM}
\end{equation}
with $m$, which is allowed to vary across sample observations.

\section{Maximum likelihood estimation: an EM algorithm}
\label{Section maximum likelihood estimation}

In this section, we present an EM algorithm for maximum likelihood (ML) estimation of the parameters of the cBB-RM (Section \ref{section em algorithm}), followed by a discussion on the initialization strategy and convergence criteria used (Section \ref{Section initialization and covergence}).

\subsection{An EM algorithm}
\label{section em algorithm}

Let $(\bx_1',\bu_1',\bv_1',\bz_1',y_1),\dots,(\bx_n',\bu_n',\bv_n',\bz_n',y_n)$ be the observed sample from \eqref{eq:cBB-RM}, with $m_1,\ldots,m_n$ denoting the maximum possible counts for each unit.
For the application of the EM algorithm, it is convenient to view the observed data as incomplete. 
In this case, the source of incompleteness stems from the fact that we do not know if the generic data point $(\bx_i',\bu_i',\bv_i',\bz_i',y_i)$ comes from the reference or contaminant BB-RM. 
To denote the source of incompleteness, we use an indicator vector $\bw=(w_1,\dots,w_n)$ so that $w_i=1$ if $(\bx_i',\bu_i',\bv_i',\bz_i',y_i)$ comes from the contaminant BB-RM and $w_i=0$ otherwise.
The complete-data are thus given by $(\bx_1',\bu_1',\bv_1',\bz_1',y_1, w_1),\dots,(\bx_n',\bu_n',\bv_n',\bz_n',y_n, w_n)$ and, from (\ref{pdf contaminated beta binomial}), the complete-data likelihood function can be written as 
\begin{align*}
    L_c(\bbeta,\balpha, \bgamma, \blambda) = & \prod^{n}_{i=1}\left[(1-\delta(\bv_i;\bgamma)) f_{\text{BB}_{m_i}}(y_i;{\pi}(\bx_i;\bbeta),{\sigma}(\bu_i;\balpha))\right]^{1-w_i}\\
    &\times\left[\delta(\bv_i;\bgamma) f_{\text{BB}_{m_i}}(y_i;{\pi}(\bx_i;\bbeta),{\eta}(\bz_i;\blambda){\sigma}(\bu_i;\balpha))\right]^{wi}\\
    =&\prod_{i=1}^n\left[\frac{(1-\delta(\bv_i;\bgamma)) m_i!}{y_i!(m_i-y_i)!}\frac{\mathrm{B}\left(y_i+\frac{{\pi}(\bx_i;\bbeta)}{{\sigma}(\bu_i;\balpha)},m_i-y_i+\frac{1-{\pi}(\bx_i;\bbeta)}{{\sigma}(\bu_i;\balpha)}\right)}{\mathrm{B}\left(\frac{{\pi}(\bx_i;\bbeta)}{{\sigma}(\bu_i;\balpha)},\frac{1-{\pi}(\bx_i;\bbeta)}{{\sigma}(\bu_i;\balpha)}\right)}\right]^{1-w_i}\\
    &\times\left[\frac{\delta(\bv_i;\bgamma) m_i!}{y_i!(m_i-y_i)!}\frac{\mathrm{B}\left(y_i+\frac{{\pi}(\bx_i;\bbeta)}{{\eta}(\bz_i;\blambda) {\sigma}(\bu_i;\balpha)},m_i-y_i+\frac{1-{\pi}(\bx_i;\bbeta)}{{\eta}(\bz_i;\blambda) {\sigma}(\bu_i;\balpha)}\right)}{\mathrm{B}\left(\frac{{\pi}(\bx_i;\bbeta)}{{\eta}(\bz_i;\blambda) {\sigma}(\bu_i;\balpha)},\frac{1-{\pi}(\bx_i;\bbeta)}{{\eta}(\bz_i;\blambda) {\sigma}(\bu_i;\balpha)}\right)}\right]^{w_i}.
\end{align*}
The complete log-likelihood function then follows as
\begin{align}
l_c(\bbeta,\balpha,\bgamma,\blambda)=l_{{c_1}}(\bgamma)+l_{c_2}(\bbeta,\balpha,\blambda)\label{eq complete loglikelihood}
\end{align}
where
\begin{align*}
    l_{c_1}(\bgamma)=\sum^n_{i=1} (1-w_i) \mathrm{log}(1-\delta(\bv_i;\bgamma))+w_i\mathrm{log}\delta(\bv_i;\bgamma)
\end{align*}
and
\begin{align*}
    &l_{c_2}(\bbeta,\balpha,\blambda)\\
    =&\sum^n_{i=1}\mathrm{log}\left(\frac{ m_i!}{y_i!(m_i-y_i)!}\right)+ (1-w_i)\left[\mathrm{log}\mathrm{B}\left(y_i+\frac{{\pi}(\bx_i;\bbeta)}{{\sigma}(\bu_i;\balpha)},m_i-y_i+\frac{1-{\pi}(\bx_i;\bbeta)}{{\sigma}(\bu_i;\balpha)}\right)-\mathrm{log}\mathrm{B}\left(\frac{{\pi}(\bx_i;\bbeta)}{{\sigma}(\bu_i;\balpha)},\frac{1-{\pi}(\bx_i;\bbeta)}{{\sigma}(\bu_i;\balpha)}\right)\right]\nonumber\\
    &+ w_i\left[\mathrm{log}\mathrm{B}\left(y_i+\frac{{\pi}(\bx_i;\bbeta)}{{\eta}(\bz_i;\blambda){\sigma}(\bu_i;\balpha)},m_i-y_i+\frac{1-{\pi}(\bx_i;\bbeta)}{{\eta}(\bz_i;\blambda){\sigma}(\bu_i;\balpha)}\right)-\mathrm{log}\mathrm{B}\left(\frac{{\pi}(\bx_i;\bbeta)}{{\eta}(\bz_i;\blambda){\sigma}(\bu_i;\balpha)},\frac{1-{\pi}(\bx_i;\bbeta)}{{\eta}(\bz_i;\blambda){\sigma}(\bu_i;\balpha)}\right)\right]\nonumber.
\end{align*}
The algorithm iterates between the E-step and M-step until convergence. These steps for the $(k+1)$th iteration of the algorithm are detailed below.
\subsubsection*{E-step}
In the E-step, the conditional expectation of the complete-data log-likelihood function is computed as 
\begin{align*}
Q\left(\bbeta,\balpha,\bgamma,\blambda|\bbeta^{(k)},\balpha^{(k)},\bgamma^{(k)},\blambda^{(k)}\right) &= Q_1\left(\bgamma|\bbeta^{(k)},\balpha^{(k)},\bgamma^{(k)},\blambda^{(k)}\right)+Q_2\left(\bbeta,\balpha,\blambda|\bbeta^{(k)},\balpha^{(k)},\bgamma^{(k)},\blambda^{(k)}\right)
\end{align*}
for the $(k+1)$-th iteration, which is in the same order as (\ref{eq complete loglikelihood}). $Q\left(\bbeta,v,\delta,\eta|\bbeta^{(k)},v^{(k)},\delta^{(k)},\eta^{(k)}\right)$ is obtained by substituting $w_i$ in \eqref{eq complete loglikelihood} by the expected a posteriori probability for a point to be an extreme value
\begin{align}
	\text{E}\left(W_i|{y}_i,\bx_i,\bu_i, \bv_i, \bz_i;\bbeta^{(k)},\balpha^{(k)},\bgamma^{(k)},\blambda^{(k)}\right)
    =
    \frac{\delta^{(k)}f_{\text{BB}_{m_i}}\left(y_i;{\pi}\left(\bx_i;\bbeta^{(k)}\right),{\eta}\left(\bz_i;\blambda^{(k)}\right) {\sigma}\left(\bu_i;\balpha^{(k)}\right)\right)}{f_{\text{cBB}_{m_i}}\left(y_i;{\pi}\left(\bx_i;\bbeta^{(k)}\right),{\sigma}\left(\bu_i;\balpha^{(k)}\right),{\delta}\left(\bv_i;\bgamma^{(k)}\right) ,{\eta}\left(\bz_i;\blambda^{(k)}\right)\right)} := w_i^{(k)}.
\end{align}

\subsubsection*{M-step}

An update $\bgamma^{(k+1)}$ for $\bgamma$ is calculated by independently maximizing  
\begin{align*}	Q_1\left(\bgamma|\bbeta^{(k)},\balpha^{(k)},\bgamma^{(k)},\blambda^{(k)}\right)=\sum_{i=1}^{n}\left\{\left(1-w_i^{(k)}\right)\mathrm{log}\left(1-{\delta}\left(\bv_i;\bgamma^{(k)}\right)\right)+w_i^{(k)}\mathrm{log}{\delta}\left(\bv_i;\bgamma^{(k)}\right)\right\}
\end{align*}
with respect to $\bgamma$ and subjects to constraints on this parameter.  It follows that an update for the $(k+1)$-th iteration is given as  
\begin{align*}
	\bgamma^{(k+1)}={\frac{1}{n}\sum_{i=1}^{n}w_i^{(k)}}.
\end{align*}
Updates for $\bbeta,\sigma$ and $\eta$ are obtained by maximizing 
\begin{align}
&Q_2\left(\bbeta,\balpha,\blambda|\bbeta^{(k)},\balpha^{(k)},\bgamma^{(k)},\blambda^{(k)}\right) \\
=&
 \sum^n_{i=1}\left(1-w_i^{(k)}\right)\left[\mathrm{log}\mathrm{B}\left(y_i+\frac{{\pi}\left(\bx_i;\bbeta^{(k)}\right)}{{\sigma}\left(\bu_i;\balpha^{(k)}\right)},m_i-y_i+\frac{1-{\pi}\left(\bx_i;\bbeta^{(k)}\right)}{{\sigma}\left(\bu_i;\balpha^{(k)}\right)}\right)-\mathrm{log}\mathrm{B}\left(\frac{{\pi}\left(\bx_i;\bbeta^{(k)}\right)}{{\sigma}\left(\bu_i;\balpha^{(k)}\right)},\frac{1-{\pi}\left(\bx_i;\bbeta^{(k)}\right)}{{\sigma}\left(\bu_i;\balpha^{(k)}\right)}\right)\right]\nonumber\\
    &+ w_i^{(k)}\left[\mathrm{log}\mathrm{B}\left(y_i+\frac{{\pi}\left(\bx_i;\bbeta^{(k)}\right)}{{\eta}\left(\bz_i;\blambda^{(k)}\right){\sigma}\left(\bu_i;\balpha^{(k)}\right)},m_i-y_i+\frac{1-{\pi}\left(\bx_i;\bbeta^{(k)}\right)}{{\eta}\left(\bz_i;\blambda^{(k)}\right){\sigma}\left(\bu_i;\balpha^{(k)}\right)}\right)\right.\nonumber\\
    &-\left.\mathrm{log}\mathrm{B}\left(\frac{{\pi}\left(\bx_i;\bbeta^{(k)}\right)}{{\eta}\left(\bz_i;\blambda^{(k)}\right){\sigma}\left(\bu_i;\balpha^{(k)}\right)},\frac{1-{\pi}\left(\bx_i;\bbeta^{(k)}\right)}{{\eta}\left(\bz_i;\blambda^{(k)}\right){\sigma}\left(\bu_i;\balpha^{(k)}\right)}\right)\right]\nonumber.
\end{align} 
This can be achieved in \texttt{R} software using the \texttt{optim()} function included in the \textbf{stats} package. 
The Nelder-Mead or BFGS algorithms, which are used for unconstrained optimization, can be passed to \texttt{optim()} via the \texttt{method} argument. 
The algorithm iterates between the E-step and M-step until convergence and is elaborated on below.
\subsection{Initialization and convergence}\label{Section initialization and covergence}

The initial values are an essential element in EM-based algorithms and can significantly influence the accuracy and reliability of the model estimates; therefore, their choice represents a vital aspect of the estimation process.
Should the initial values be selected inadequately, the algorithm might settle at a local maximum rather than achieving the global maximum. Furthermore, if the initial values significantly diverge from the actual values, the algorithm may experience slow convergence. 

We recommend applying a standard BB-RM utilizing the same predictors as those used for cBB-RM in the data analysis. After that, the computed coefficients can be used as starting points for cBB-RM fitting.
For $\delta$ and $\eta$, we suggest choosing them such that the cBB-RM tends to the BB-RM, i.e., $\delta^{(0)}\to1^-$ and $\eta^{(0)}\to1^+$ (see 
 Proposition \ref{proposition 1}).
 
Several convergence criteria can be applied to determine whether or not the EM algorithm has converged. A prevalent approach involves monitoring the variation in the log-likelihood function across successive iterations.
If the change falls below a predetermined threshold, the algorithm can be considered to be converged, i.e. $l(\bbeta^{(k+1)},\balpha^{(k+1)},\bgamma^{(k+1)},\blambda^{(k+1)})-l(\bbeta^{(k)},\balpha^{(k)},\bgamma^{(k)},\blambda^{(k)})<\epsilon$. Due to the possibility of flat likelihoods, we chose a stopping criteria of $\epsilon=1 \times 10^{-10}$ or $1000$ iterations.

\section{Simulation study: sensitivity analysis} \label{Section simulation}

In this study, we performed a sensitivity analysis to examine how extreme observations affect the estimates of the B-RM, BB-RM, and cBB-RM. 
We generate 1000 datasets of sizes $n=500$ from the B-D, with $m=10$, an intercept of $\beta_0=2$, and a continuous covariate generated by a uniform distribution over the interval $(0,1)$ with slope $\beta_1=1$. 
The generated data is then augmented with extreme values using one of the following schemes:
\begin{enumerate}
    \item {$1\%$} of the generated $Y$-values are randomly substituted by data generated from a uniform distribution over the interval $(0,10)$.
    \item {$5\%$} of the generated $Y$-values are randomly substituted by data generated from a uniform distribution over the interval $(0,10)$.
\end{enumerate}
Examples of the aforementioned schemes are illustrated in \figurename~\ref{plot sensitivity example}, where the artificially modified observations are highlighted in red.  
The B-RM, BB-RM, and cBB-RN models are then fitted to the data to assess the impact of the red observations on the estimated regression coefficients.  
The bias and mean squared error (MSE), reported in \tablename~\ref{table sensitivity analysis}, are computed as
\begin{align*}
\text{Bias}(\hat\beta_h)=\left(\frac{1}{1000}\sum_{r=1}^{1000}\hat{\beta}_{hr}\right)-\beta_h
\end{align*}
and
\begin{align*}
    \text{MSE}(\hat\beta_h)=\frac{1}{1000}\sum_{r=1}^{1000}\left(\hat{\beta}_{hr}-\beta_h\right)^2,
\end{align*}
where $\hat{\beta}_{hr}$ is the estimate of $\beta_h$, $h=0,1$, obtained at the $r$th replication $r=1,\ldots,1000$.

\begin{figure}[!ht]
	\centering
	\begin{subfigure}[h]{0.48\textwidth}
		\centering
  \includegraphics[scale=0.48]{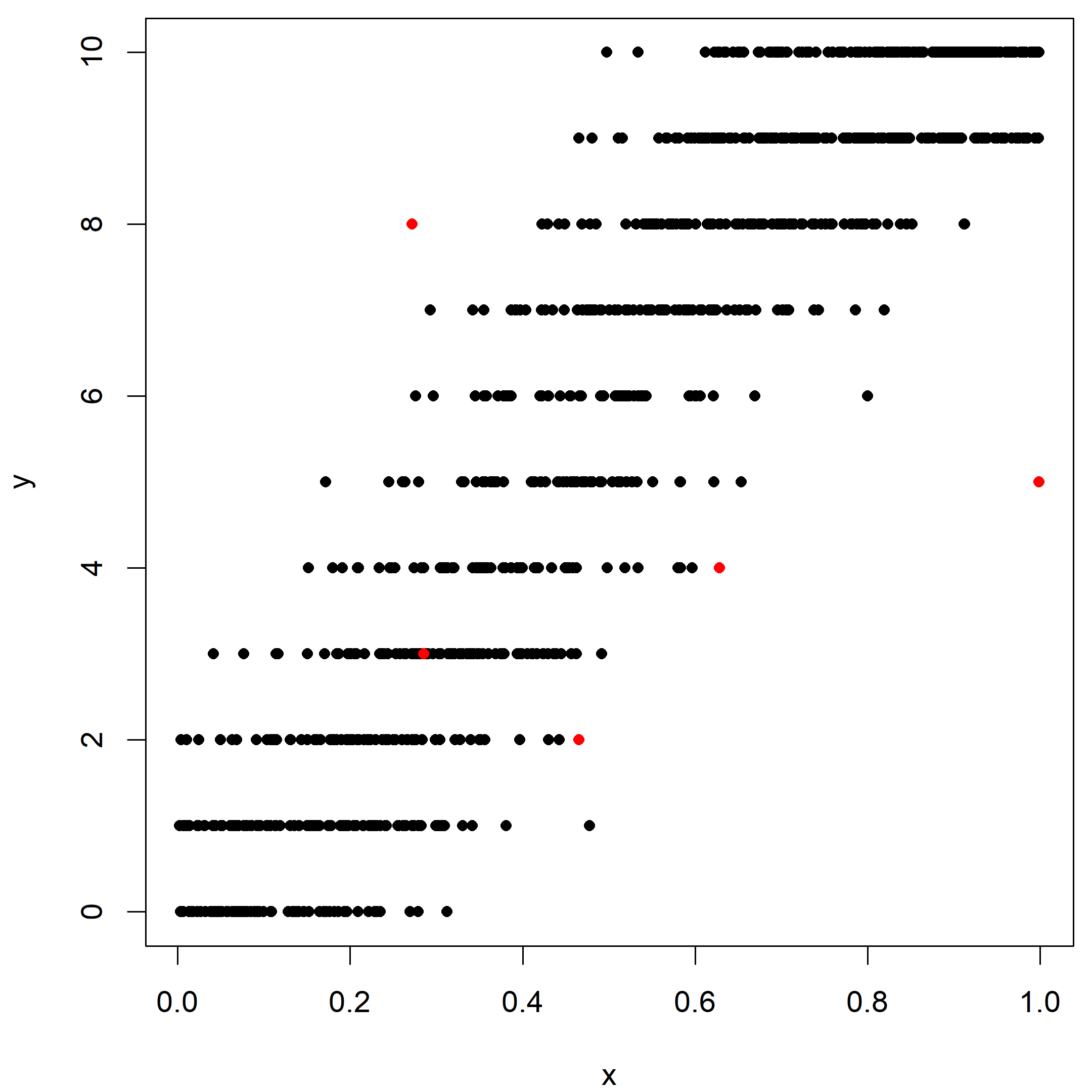}
		\caption{1\%.}
	\end{subfigure}
 \begin{subfigure}[h]{0.48\textwidth}
		\centering
		\includegraphics[scale=0.48]{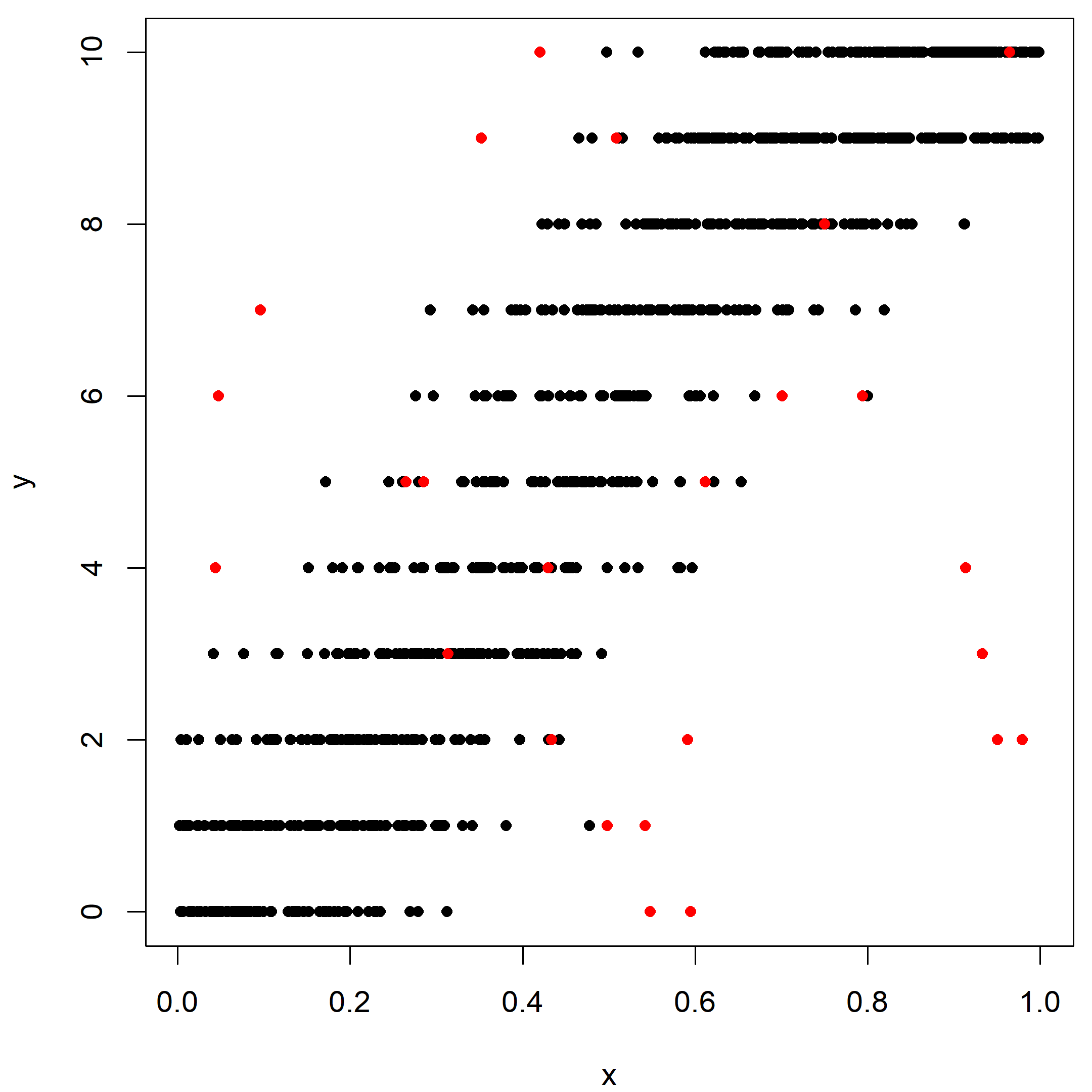}
		\caption{5\%. }
	\end{subfigure}
	\caption{Examples of simulated binomial data with a different percentage of artificially added outliers (in red).}
\label{plot sensitivity example}
\end{figure}

\begin{table}[h!]
\centering
\caption{Sensitivity analysis to examine how the artificially added anomalous values affect the estimates of the regression models.}
\label{table sensitivity analysis}
\begin{tabular}{clclrrrlrrr}
\toprule
                      &  & \multicolumn{1}{l}{} &  & \multicolumn{3}{c}{1\%}                                                        &                      & \multicolumn{3}{c}{5\%}                                                        \\ \cline{5-7} \cline{9-11} 
                      &  & \multicolumn{1}{l}{} &  & \multicolumn{1}{c}{B-RM} & \multicolumn{1}{c}{BB-RM} & \multicolumn{1}{c}{cBB-RM} &                      & \multicolumn{1}{c}{B-RM} & \multicolumn{1}{c}{BB-RM} & \multicolumn{1}{c}{cBB-RM} \\ \midrule
\multirow{2}{*}{Bias} &  & $\hat{\beta}_0$            &  & 0.0605                    & 0.0536                 & 0.0314                  & \multicolumn{1}{r}{} & 0.3127                    & 0.2809                 & 0.1625                  \\
                      &  & $\hat{\beta}_1$            &  & -0.1477                   & -0.1287                & -0.0753                 & \multicolumn{1}{r}{} & -0.7526                   & -0.6638                & -0.3840                 \\
                      &  &                      &  &                           &                        &                         & \multicolumn{1}{r}{} &                           &                        &                         \\
\multirow{2}{*}{MSE}  &  & $\hat{\beta}_0$            &  & 0.0122                    & 0.0112                 & 0.0088                  &                      & 0.1100                    & 0.0897                 & 0.0347                  \\
                      &  & $\hat{\beta}_1$            &  & 0.0584                    & 0.0513                 & 0.0382                  &                      & 0.6200                    & 0.4861                 & 0.1826                  \\ \bottomrule
\end{tabular}
\end{table}

The sensitivity analysis reveals that the presence of outliers significantly affects the estimates of the B-RM as reflected by the bias and MSE. Although the BB-RM performs better than the B-RM, it is also sensitive to outliers. 
This is even more pronounced as the percentage of outliers increases from 1\% to 5\%, as expected.
The cBB-RM consistently demonstrates the lowest bias and MSE for both regression coefficients, indicating that the cBB-RM is more reliable for handling data prone to contamination.

\section{Real data applications}
\label{Section application} 

In this section, we investigate the behaviour of the cBB-D and cBB-RM by applying them to real-world datasets, namely survival data of mule deer due to winter malnutrition (Section \ref{section mule deer}) and the Eurobarometer 95.1 survey (Section \ref{section eurobarometer}). 
To illustrate the model's viability as an alternative for overdispersed bounded counts, we benchmark it to a flexible BB generalization known as the beta-2-binomial (B2B) model \citep{bayes2024robust}.
Model performance is ranked via the  Akaike information criterion (AIC; \citealp{akaike1974new}),
\begin{align}
    \text{AIC}=2k-2l(\hat{\bbeta},\hat{\balpha}, \hat{\bgamma}, \hat{\blambda}),
\end{align}
the Bayesian information criterion (BIC; \citealp{schwarz1978estimating}),
\begin{align*}
    \text{BIC} = \mathrm{log}(n)k-2l(\hat{\bbeta},\hat{\balpha}, \hat{\bgamma}, \hat{\blambda}) ,
\end{align*}
and the Hannan-Quinn information criterion (HQIC; \citealp{hannan1979determination})
\begin{align*}
    \text{HQIC} = 2\mathrm{log}(\mathrm{log}(n))k-2l(\hat{\bbeta},\hat{\balpha}, \hat{\bgamma}, \hat{\blambda}), 
\end{align*}
where $k$ is the number of parameters, $n$ is the number of observations, 
$\hat{\bbeta},\hat{\balpha}, \hat{\bgamma}$, and  $\hat{\blambda}$ are the ML estimates of $\bbeta, \balpha,\bgamma$ and $\blambda$, respectively, and  where $l(\hat{\bbeta},\hat{\balpha}, \hat{\bgamma}, \hat{\blambda})$ is the maximized log-likelihood value for the cBB-RM.

Moreover, we use the likelihood-ratio (LR) test, which evaluates nested models, to assess whether the cBB-RM (alternative model) provides a significant improvement over the BB-RM (null model), as the latter is a special case of the former. 
Under the null hypothesis of no improvement, the test statistic is 
\begin{align} 
\label{eq likelihood ratio}
\text{LR}=-2\left[l(\hat{\bbeta}, \hat{\balpha}) -l(\hat{\bbeta},\hat{\balpha}, \hat{\bgamma}, \hat{\blambda})\right],
\end{align} 
 where $l(\hat{\bbeta}, \hat{\balpha})$ is the maximized log-likelihood value for the BB-RM and $l(\hat{\bbeta},\hat{\balpha}, \hat{\bgamma}, \hat{\blambda})$ is as previously defined. 
Based on Wilk's theorem, the LR statistic approximately follows a $\chi^2$ distribution with degrees of freedom equal to the number of parameters between the alternative and null models. 
This enables the calculation of a $p$-value to determine the significance of the improvement.

After using the EM algorithm outlined in Section \ref{section em algorithm} to estimate the model parameters, we obtain the variance-covariance matrix of the parameter estimates by inverting the negative Hessian matrix, which is computed using the \texttt{optim()} function in \texttt{R}. 
The standard errors of the cBB-RM parameter estimates are then determined by taking the square roots of the diagonal elements of this matrix.

\subsection{Mule Deer Mortality} 
\label{section mule deer}

The mortality rates of mule deer (\textit{Odocoileus hemionus}) fawns due to winter malnutrition were analyzed using data collected from radio-collared fawns captured during early winter in Colorado, Idaho, and Montana in the United States of America between 1981 and 1996 \citep{unsworth1999mule}. 
The study consisted of $26$ separate observations encompassing a total of $n=1875$ radio-collared mule deer collected over the years in the three states.
This dataset represents a comprehensive, long-term study of overwinter survival in a species that is highly sensitive to environmental conditions. 

As a first step in the analysis, we fit an intercept-only cBB-RM to \textsf{deaths} ($Y$), representing the number of mule deer that died to winter malnutrition out of a total of $m_i$ radio-collared fawns, for $i = 1, \dots, 26$.
This is equivalent to fitting the cBB-D, but with varying $m$. 
\tablename~\ref{table mule cBB-d AIC BIC} presents the results of comparing the cBB-D to the B-D, BB-D, and B2B-D.
The models are separately ranked via their AIC, BIC, and HQIC values.
The cBB-RM performed the best, but it is also worth noting that the binomial distribution is not enough for this dataset, regardless of the considered measure. 
The LR test further supports that the cBB-RM is an improvement over the BB-RM with a $p$-value of $0.007$.
As for the cBB-RM, it is worth noting that the estimate of the contamination proportion is $\widehat{\delta}=0.603$, showing a large excess of extreme values with respect to the reference cBB-RM for the regular counts.
\begin{table}[!ht]
	\centering
	\caption{Ranking of fitted models to the mule deer mortality data according to the AIC, BIC, and HQIC.}
	
	\begin{tabular}{l r r c c c c c c } 
		\toprule
		Model & \#par & log-likelihood & AIC & rank & BIC & rank & HQIC & rank\\
		
		\midrule

B-RM	&	1	& -234.424	& 470.849	& 4	& 472.107	&	4 & 471.211 &4\\
BB-RM	&	2	& -90.125	&	184.249& 3	&186.765	&2 &184.974 &3	\\
cBB-RM	&	4	&-86.470	&180.940	&1	&185.973	&1	&182.390 &1\\
B2B-RM	&	3	&-88.832	&183.663	&2	&187.438	&3	&184.750&2\\

		\bottomrule
	\end{tabular}
	\label{table mule cBB-d AIC BIC}
\end{table}


In the second part of the analysis, we account for the fact that fawn survival outcomes are influenced by multiple factors, including regional variations. Specifically, we now model $\pi$, which corresponds to the average number of deaths through the relationship in \eqref{ec:cBB mean}, as a (generalized) linear function of the nominal covariate \textsf{state} ($X$), indicating the location where the fawns were observed (Colorado, Idaho, or Montana).  
The cBB-RM is specified as:
\begin{align*}
    \textsf{deaths}_i|\textsf{state}_i &\sim c\mathcal{BB}_{m_i}(\pi(\textsf{state}_i;\bbeta),\sigma,\delta,\eta)\\
    \logit(\pi(\textsf{state}_i;\bbeta)) &= \beta_0 + \beta_{\textsf{Idaho}} I(\textsf{state}_i=\textsf{Idaho})+\beta_{\text{Montana}} I(\textsf{state}_i=\text{Montana})\\
    \log(\sigma)&=\alpha_0\\
    \logit(\delta)&=\gamma_0\\
    \log(\eta-1)&=\lambda_0
\end{align*}
where $I(A)$ represents an indicator (dummy) variable that takes the value 1 if condition $A$ is met and 0 otherwise, with \textsf{Colorado} serving as the reference category for \textsf{state}.
As shown in \tablename~\ref{table mule cBB-rm AIC BIC}, the cBB-RM outperforms the other regression models according to the AIC and BIC.

The likelihood ratio (LR) test confirms that the cBB-RM model provides a statistically significant improvement over the BB-RM model, with a $p$-value of 0.001.
Additionally, an LR test can be performed to compare the (null) cBB-D model against the (alternative) cBB-RM model, as the former represents a special case of the latter when $\beta_1 = \beta_2 = 0$. 
The resulting $p$-value of 0.007 indicates that incorporating state membership to model the mean count within the cBB framework leads to a significant improvement in model fit.
\begin{table}[!ht]
	\centering
	\caption{Ranking of fitted models to the mule deer mortality data with state as a covariate according to the AIC, BIC, and HQIC.}	
	\begin{tabular}{l r r c c c c c c } 
		\toprule
		Model & \#par & log-likelihood & AIC & rank & BIC & rank & HQIC & rank\\		
		\midrule
B-RM	&	3	& -224.930	&445.860 &	4&459.634 &4 & 456.947 &4	\\
BB-RM	&	4	& -88.103	& 184.205& 3 & 189.238&3 & 185.654 &3	\\
cBB-RM	&	6	& -81.546&	175.092& 1&182.640 &1 & 177.265 & 2	\\
B2B-RM	&	5	&-84.790 & 179.581	&2 &185.871 &2	& 176.668 &1\\
		\bottomrule
	\end{tabular}
	\label{table mule cBB-rm AIC BIC}
\end{table}

\tablename~\ref{table mule deer survavl est RM} shows the estimated regression coefficients for the BB and cBB regression models, along with their standard errors in round brackets.
Since a logit link function is used, the coefficients describe the change in log-odds of a mule deer dying from winter malnutrition. 
In the BB-RM, the intercept $\hat{\beta}_0 = -1.186$ (SE = 0.267) represents the log-odds of death in Colorado, corresponding to an odds of $e^{-1.186} \approx 0.31$ or probability of $\frac{0.31}{1+0.31} \approx 0.24$. 
The coefficient for Idaho is $\hat{\beta}_1 = -1.010$ (SE = 0.634), indicating that mule deer in Idaho have lower log-odds of death compared to Colorado, with an odds ratio of $e^{-1.010} \approx 0.36$, meaning about 64\% lower odds of death. 
However, the standard error is large, suggesting some uncertainty. 
The coefficient for Montana is $\hat{\beta}_2 = -0.774$ (SE = 0.560), implying an odds ratio of $e^{-0.774} \approx 0.46$, or about 54\% lower odds of death compared to Colorado, though with a notable standard error. 
In the cBB-RM, the intercept $\hat{\beta}_0 = -0.968$ (SE = 0.114) is slightly higher in absolute value, and the smaller standard error suggests greater precision. 
The Idaho coefficient $\hat{\beta}_1 = -1.603$ (SE = 0.555) indicates a stronger reduction in log-odds compared to the BB-RM, with an odds ratio of $e^{-1.603} \approx 0.20$, meaning Idaho mule deer are about 80\% less likely to die than those in Colorado. 
The Montana coefficient $\beta_2 = -0.281$ (SE = 0.333) is much smaller than in the BB model, corresponding to an odds ratio of $e^{-0.281} \approx 0.76$, meaning only a 24\% lower odds of death. 
The standard error for Montana is also smaller, suggesting a more stable estimate. Overall, the cBB-RM suggests a stronger effect for Idaho and a weaker effect for Montana compared to the BB model, while also providing more precise estimates, likely due to better handling of extreme values and data heterogeneity.
Specifically, since $\hat{\delta} = \logit^{-1}(\hat{\gamma}_0) = 0.567$, approximately 56.7\% of observations belong to the contaminant BB component, indicating an excess of extreme counts relative to the reference BB-RM. 
This percentage is lower than the 60.3\% obtained in the initial analysis without covariates, suggesting that incorporating the state variable to explain the mean count reduces some of the excess variability previously attributed to extreme observations. 
In other words, state membership accounts for part of the overdispersion observed when no covariates were included.  
Furthermore, the estimated degree of contamination parameter is given by $\hat{\eta} = e^{\hat{\lambda}_0} + 1 = e^{4.606} + 1 \approx 100.8$, indicating that counts from the contaminant BB component are about 100 times more dispersed than those from the reference BB component. 
This suggests that while state membership helps explain some of the extreme observations, a substantial degree of contamination remains, emphasizing the necessity of using a flexible model like cBB-RM to properly capture the data heterogeneity.
\begin{table}[!ht]
\centering
\caption{Estimated coefficients and corresponding SEs (in brackets) of BB and cBB regression models to mule deer mortality data.}
\label{table mule deer survavl est RM}
\begin{tabular}{lrrlrr}
\toprule
& \multicolumn{5}{c}{Estimates (SEs)}\\
Parameter & \multicolumn{2}{c}{BB}                      &                      & \multicolumn{2}{c}{cBB} \\ 
\midrule
$\beta_0$   &    -1.186           &   (0.267)             &                      & -0.968       & (0.114)     \\
$\beta_{\text{Idaho}}$   &    -1.010           &   (0.634)             &                      & -1.603       & (0.555)     \\
$\beta_{\text{Montana}}$   &    -0.774           &   (0.560)             &                      & -0.281       & (0.333)     \\

$\alpha_0 $ & -1.427               & (0.325)               &                      & -5.281     & (1.373)     \\
$\gamma_0 $ &  & &                      & 0.269      & (0.578)     \\
$\lambda_0$ &  &  &                      & 4.606       & (1.457)     \\ 
\bottomrule
\end{tabular}
\end{table}

\subsection{Eurobarometer 95.1 Survey} 
\label{section eurobarometer}

Since the early 1970s, the European Commission has regularly surveyed public opinion in the Member States of the European Union using the "Standard and Special Eurobarometer". 
The Eurobarometer 95.1 survey was conducted in March--April 2021. 
For this analysis, we focus on responses from $n=998$ survey participants in the Netherlands. 
The dataset is freely accessible at \citet{Eurobarometer}. 
Our response variable $Y_i$ is the response of the seriousness of climate change out of 10, which we shifted by 1 so that '0' is instead considered as 'not at all a serious problem' and '9', which is the $m$ of the model, as 'an extremely serious problem' \citep{sciandra2024discrete}.
Hereafter, we shall also refer to $Y$ as \textsf{severity}.
A bar plot displaying the absolute frequency distribution of the observed $Y$-values is shown in \figurename~\ref{fig:bar plot survey}.
Given the availability of demographic and socioeconomic data, we also analyse responses in relation to specific factors. 
In \text{M}$_1$, we begin by fitting the data without any covariates to investigate the suitability of the cBB distribution. 
This initial approach assesses the responses on the perceived severity of climate change, providing a quantitative measure of individual concerns. 
In \text{M}$_2$, we include an explanatory variable on $\pi$, denoted as ``politics", representing the political ideology of the respondents, measured by a scale where respondents rated themselves from 1 (``left") to 10 (``right"), while keeping the dispersion, proportion of extreme points and degree of contamination constant.
\text{M}$_3$ is extended by including ``age"  as an explanatory variable on the overdispersion, while \text{M}$_4$ includes the settlement size where the respondent lives (either \textsf{rural area}, \textsf{small town}, or \textsf{large town}) as an explanatory variable of the degree of contamination. 
In formulas, the models are specified below.
\begin{description}[font=\normalfont]
    \item[$\text{M}_1$:]
    \begin{align*}
    \textsf{severity}_i & \sim c\mathcal{BB}_9(\pi,\sigma,\delta,\eta)\\
    \logit(\pi) &= \beta_0\\
    \log(\sigma)&= \alpha_0\\
    \logit(\delta)&=\gamma_0\\
    \log(\eta-1)&=\lambda_0
\end{align*}
\item[$\text{M}_2$:]
    \begin{align*}
    \textsf{severity}_i|\textsf{politics}_i & \sim c\mathcal{BB}_9(\pi(\textsf{politics}_i;\bbeta),\sigma,\delta,\eta)\\
    \logit(\pi(\textsf{politics}_i;\bbeta)) &= \beta_0+\beta_1\text{politics}_i \\
    \log(\sigma)&= \alpha_0 \\
    \logit(\delta)&= \gamma_0 \\
    \log(\eta-1)&= \lambda_0
\end{align*}
\item[$\text{M}_3$:]
    \begin{align*}
    \textsf{severity}_i|\textsf{politics}_i,\text{age}_i & \sim c\mathcal{BB}_9(\pi(\textsf{politics}_i;\bbeta),\sigma(\textsf{age}_i;\balpha),\delta,\eta)\\
    \logit(\pi(\textsf{politics}_i;\bbeta)) &= \beta_0+\beta_1\textsf{politics}_i \\
    \log(\sigma(\textsf{age}_i;\balpha))&= \alpha_0+\alpha_1 \textsf{age}_i \\
    \logit(\delta)&= \gamma_0 \\
    \log(\eta-1)&= \lambda_0
\end{align*}
\item[$\text{M}_4$:]
    \begin{align*}    \textsf{severity}_i|\textsf{politics}_i,\text{age}_i,\textsf{settlement}_i & \sim c\mathcal{BB}_9(\pi(\textsf{politics}_i;\bbeta),\sigma(\textsf{age}_i;\balpha),\delta,\eta(\textsf{settlement}_i;\blambda))\\
    \logit(\pi(\textsf{politics}_i;\bbeta)) &= \beta_0+\beta_1\text{politics}_i \\
    \log(\sigma(\textsf{age}_i;\balpha))&= \alpha_0+\alpha_1 \text{age}_i \\
    \logit(\delta)& = \gamma_0 \\
    \log(\eta(\textsf{settlement}_i;\blambda)-1)& = \lambda_0+\lambda_1I(\textsf{settlement}_i=\textsf{rural})+\lambda_2I(\textsf{settlement}_i=\text{small town})
\end{align*}
\end{description}
for $i=1,\dots,998$, with \textsf{large town} serving as the reference category for \textsf{settlement} size.

\begin{figure}[!ht]
    \centering
    \includegraphics[scale=0.5]{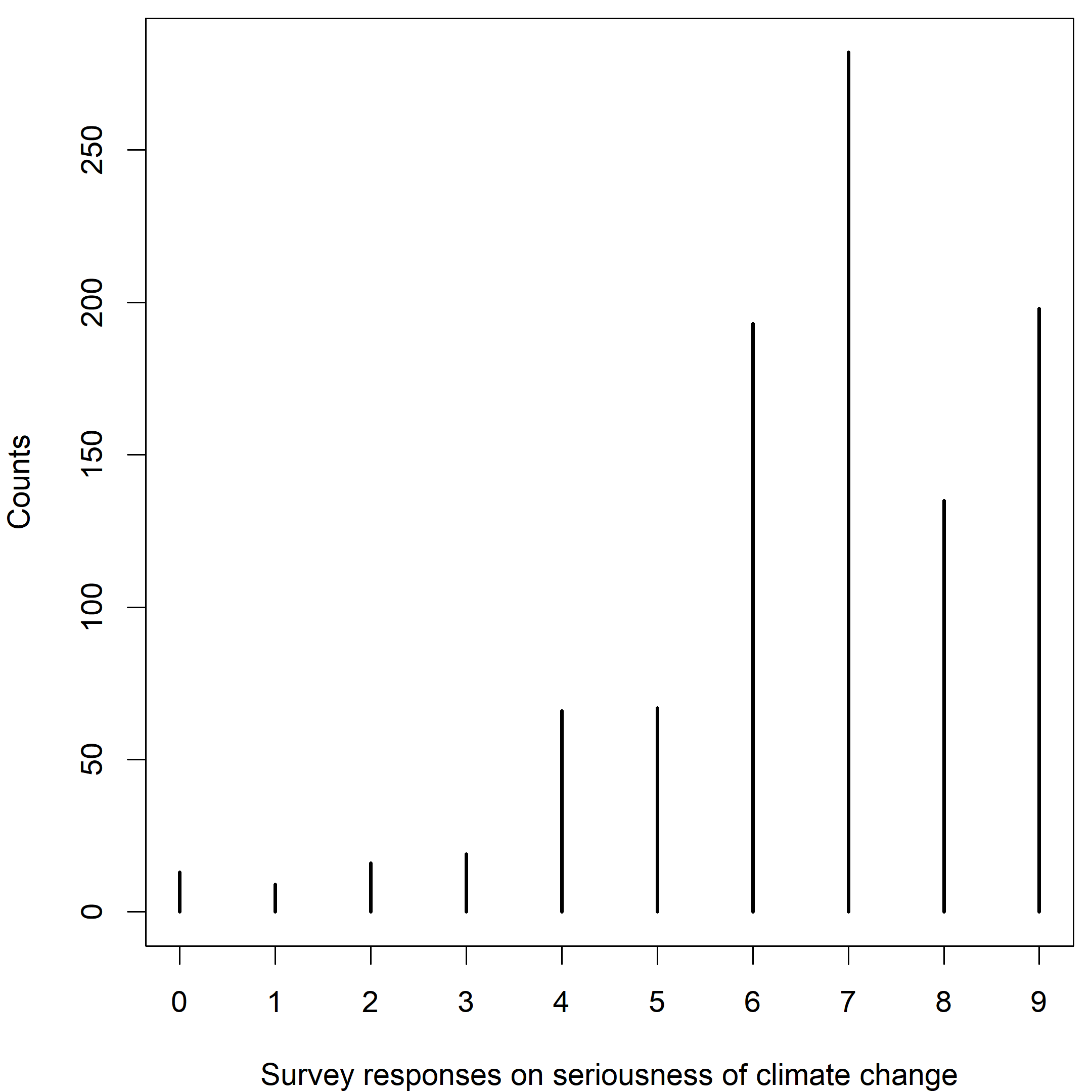}
    \caption{Barplot of the counts of the survey responses on seriousness of climate change (0 = Not at all serious, 9 = Extremely serious)".}
    \label{fig:bar plot survey}
\end{figure}

\begin{table}[!ht]
	\centering
	\caption{Ranking of fitted cBB-RMs to the Netherlands Survey Responses on Climate Change data according to the AIC, BIC and HQIC.}
	
	\begin{tabular}{l r r c c c c c c  } 
		\toprule
		Model & \#par & log-likelihood & AIC & rank & BIC & rank & HQIC& rank\\
		
		\midrule

\text{M}$_1$	&	4	& -1890.949		&	3789.899	& 4	&3809.522		& 4	& 3797.357 & 4 \\
\text{M}$_2$	& 5 & -1809.903 & 3629.807&3 & 3654.336 &	3& 3639.130 & 3\\
\text{M}$_3$	& 6 & -1805.580& 3623.159 &2 &3652.594	& 1& 3634.347&1 \\
\text{M}$_4$& 8 &-1803.433 & 3622.866&1 &3662.112 &	2&3637.783 &2 \\

        \bottomrule
	\end{tabular}
	\label{table NL survey cbb AIC BIC HQIC}
\end{table}

\begin{table}[!ht] 
    \caption{Pairwise likelihood ratio test $p$-values of fitted cBB-RM models to the Netherlands Survey responses on Climate Change data.}
	
    \begin{tabular}{|c|ccc c|}
    \toprule
        & M$_1$ & M$_2$ & M$_3$ &  M$_4$ \\
        \midrule
        M$_1$ &&&&\\
        M$_2$ & 0.000 &  & & \\
        M$_3$ & 0.000 & 0.003 & & \\
        M$_4$ & 0.000 & 0.004 & 0.117 &\\
        \bottomrule
    \end{tabular}
    \label{table NL LR p-values}
    \centering
\end{table}

As shown in \tablename~\ref{table NL survey cbb AIC BIC HQIC}, M$_4$ performed the best based on the AIC, whereas M$_3$ did according to the BIC  and HQIC.  
This is supported by the pairwise LR test $p$-values in \tablename~\ref{table NL LR p-values}, which shows that M$_4$ is not a statistically significant improvement over M$_3$.
Nevertheless, the parameter estimates and SEs for both M$_3$ and M$_4$ are given in \tablename~\ref{table NL survey est M3 M4}. 


\begin{table}[!ht]
\centering
\caption{Estimated coefficients and corresponding SEs (in brackets) of BB-RM and cBB-RM to the Netherlands Survey data, using covariates as specified in M$_3$ and M$_4$.}
\label{table NL survey est M3 M4}
\begin{tabular}{lrrlrr}
\toprule
Parameter & \multicolumn{2}{c}{M$_3$}                      &                      & \multicolumn{2}{c}{M$_4$} \\ 
\midrule
$\beta_0$   &   2.146              & (0.095)                &                      &   2.131     & (0.094)     \\
$\beta_1 $  &   -0.213             &    (0.016)             &  & -0.210      & (0.016)     \\
$\alpha_0 $ & -6.997              & (1.376)               &                      & -7.740      & (1.483)     \\
$\alpha_1 $ & 0.045              & (0.017)               &                      & 0.057      & (0.021) \\

$\gamma_0 $ & -1.450 & (0.317)&                      & -1.499      & (0.274)     \\
$\lambda_0$ & 4.840 & (0.820) &                      &   6.131    & (1.130)\\

$\lambda_1$ & & &                      &   -1.435    & (0.934)\\

$\lambda_2$ & & &                      &   -1.343    & (0.850)\\
\bottomrule
\end{tabular}
\end{table}

Due to the different parametrizations of the BB-RM, cBB-RM, and B2B model, direct comparisons for M$_3$ and M$_4$ are not possible, as certain parameters (e.g., the contamination degree $\eta$) are absent in the BB-RM and B2B model. However, comparisons can be made for M$_1$ and M$_2$, as seen in Tables \ref{table NL survey AIC BIC} and \ref{table NL survey mean AIC BIC}, respectively. 

\begin{table}[h!]
	\centering
	\caption{Ranking of fitted models to the Netherlands Survey responses on Climate Change data according to the AIC, BIC and HQIC, using no covariates as specified in M$_1$.}
	
	\begin{tabular}{l r r c c c c c c } 
		\toprule
		Model & \#par & log-likelihood & AIC & rank & BIC & rank & HQIC &rank\\
		
		\midrule

B-D	&	1	&	-2133.422	&	4268.843	&4	&	4273.749	&4&4270.708&	4\\
BB-D	&	2	&	-1937.026	&	3878.051	&	3&	3887.863	&	3&3881.781&3\\
cBB-D	&	4	&	-1890.949	&	3789.899	&1	&	3809.522	&	2& 3797.357&1\\
B2B-D	&	3	&	-1893.708	&	3793.416	&	2&	3808.133	&	1&3799.010&2\\

		\bottomrule
	\end{tabular}
	\label{table NL survey AIC BIC}
\end{table}

\begin{table}[h!]
	\centering
	\caption{Ranking of fitted models to the Netherlands Survey responses on Climate Change data according to the AIC, BIC and HQIC, using covariates as specified in M$_2$.}
	
	\begin{tabular}{l r r c c c c c c  } 
		\toprule
		Model & \#par & log-likelihood & AIC & rank & BIC & rank & HQIC & rank\\
        \midrule
B-RM	&	2	&	-1951.782	&	3907.564	&	4&	3917.375	&4	&3911.293&4\\
BB-RM	&	3	&	-1843.396	&	3692.792	&	3&	3707.509	&	3&3698.386&3\\
cBB-RM	&	5	&	-1809.903	&	3629.807	&	2&	3654.336	& 2&3639.130&	2\\
B2B-RM	&	4	&	-1809.377	&	3626.755	&	1&	3646.378	& 1&3630.349& 1\\
		\bottomrule
	\end{tabular}
	\label{table NL survey mean AIC BIC}
\end{table}

\begin{table}[h!]
\centering
\caption{Estimated coefficients and corresponding SEs (in brackets) of BB and cBB regression models to the Netherlands Survey data, using covariates as specified in M$_2$.}
\label{table NL survey mean est}
\begin{tabular}{lrrlrr}
\toprule
Parameter & \multicolumn{2}{c}{BB}                      &                      & \multicolumn{2}{c}{cBB} \\ \midrule
$\beta_0$   & 2.281                & (0.096)                &                      & 2.132       & (0.101)     \\
$\beta_1 $  & -0.235               & (0.017)                &  & -0.210      & (0.018)     \\
$\alpha_0 $ & -2.239               & (0.105)               &                      & -5.406      & (4.560)     \\
$\gamma_0 $ &  & &                      & -1.271      & (0.512)     \\
$\lambda_0$ &  &  &                      & 5.430       & (4.168)     \\ \bottomrule
\end{tabular}
\end{table}

From \tablename~\ref{table NL survey AIC BIC}, we observe that the cBB-RM performs best according to the AIC and HQIC, whereas the B2B model performs better according to the BIC. However, this is not the case for M$_4$, where the B2B model outperformed the cBB-RM according to all the criteria. Despite this, the cBB-RM retains the advantage of having interpretable parameters. 

\tablename~\ref{table NL survey mean est} presents the estimated regression coefficients under $\text{M}_2$ for the BB-RM and cBB-RM, along with their standard errors in parentheses. Since a logit link function is used for the parameter $\pi$, the coefficients describe the change in log-odds of the outcome variable as a function of political ideology.
The estimated coefficient for political ideology is $\hat{\beta}_1 = -0.235$ (SE = 0.017), indicating that as respondents move one unit to the right on the political scale, their log-odds of the outcome decrease. 
This corresponds to an odds ratio of $e^{-0.235} \approx 0.79$, meaning that for each one-point increase in political ideology (toward a more right-leaning stance), the odds of the outcome decrease by approximately 21\%. 
The relatively small SE suggests a precise estimate.
In the cBB-RM, the intercept $\hat{\beta}_0 = 2.132$ (SE = 0.101) is slightly lower, suggesting a similar but slightly reduced baseline log-odds. The political ideology coefficient $\hat{\beta}_1 = -0.210$ (SE = 0.018) is also slightly smaller in magnitude compared to the BB-RM, with an odds ratio of $e^{-0.210} \approx 0.81$, meaning a slightly weaker association between political ideology and the outcome. The SE remains small.
The dispersion parameter $\alpha_0$ in the BB-RM is estimated at $-2.239$ (SE = 0.105), whereas in the cBB model, it is $-5.406$ (SE = 4.560), indicating a large shift with greater uncertainty in the contaminated model. 
The estimated proportion of observations belonging to the contaminant BB component is given by $\hat{\delta} = \logit^{-1}(\hat{\gamma}_0) = \logit^{-1}(-1.271) \approx 0.22$, meaning about the 22\% of excess of extreme responses.
The degree of contamination is estimated as $\hat{\eta} = e^{\hat{\lambda}_0} + 1 = e^{5.430} + 1 \approx 229.15$, indicating that the contaminant BB component exhibits a dispersion more than 200 times greater than that of the reference BB component. 
This suggests that while political ideology explains some of the variability in responses, a significant amount of overdispersion remains, reinforcing the importance of using a flexible model like cBB-RM to accommodate extreme values.

As noted in \tablename~\ref{table NL survey est M3 M4}, under $\text{M}_3$, the estimated mean parameters remain close to those under $\text{M}_2$, indicating a consistent relationship between political ideology and the outcome. However, the SE of the mean coefficients decreases compared to those under $M_2$. The dispersion parameter now varies with age: $\hat{\alpha}_0 = -6.997$ (SE = 1.376) and $\hat{\alpha}_1 = 0.045$ (SE = 0.017). Since dispersion is modelled on the log scale, this suggests that for each additional year of age, dispersion increases by a multiplicative factor of $e^{0.045} \approx 1.05$, meaning older respondents exhibit slightly higher overdispersion. The SE for age is small, indicating a stable estimate.
In $M_4$, settlement size is introduced as a predictor of the degree of contamination while keeping the dispersion-age relationship from $M_3$. The estimated mean parameters remain similar, and the dispersion parameters are also close to those in $M_3$, with $\hat{\alpha}_0 = -7.740$ (SE = 1.483) and $\hat{\alpha}_1 = 0.057$ (SE = 0.021), indicating a slightly stronger effect of age on dispersion.
However, settlement size has a notable effect on the degree of contamination. The baseline degree of contamination is estimated as $\hat{\lambda}_0 = 6.131$ (SE = 1.130), corresponding to $\hat{\eta} = e^{6.131} + 1 \approx 460.2$, indicating substantial contamination in the reference category (large towns). For rural respondents, $\hat{\lambda}_1 = -1.435$ (SE = 0.934), leading to a degree of contamination of $\hat{\eta}_{\text{rural}} = e^{6.131 - 1.435} + 1 \approx 139.7$, while for small-town respondents, $\hat{\lambda}_2 = -1.343$ (SE = 0.850), giving $\hat{\eta}_{\text{small town}} = e^{6.131 - 1.343} + 1 \approx 158.5$. This suggests that contamination is substantially lower in rural and small-town areas compared to large towns.

Overall, these findings highlight the benefits of the cBB framework in handling overdispersion and an excess of extreme observations. 
The inclusion of covariates for dispersion and contamination enhances model flexibility and provides deeper insights into the factors influencing the response variability.

\section{Conclusion}
\label{Section conclusion}

In this paper, we introduce the contaminated beta-binomial distribution (cBB-D) as a new approach for modeling bounded count data. 
Our model is formulated as a simple mixture of two beta-binomial distributions (BB-Ds) that share the same mean but differ in dispersion; advantageously, this implies a closed-form expression of the probability mass function. 
The primary motivation behind this formulation is to protect the reference BB-D—characterized by lower dispersion—from model misspecification, particularly due to an excess of extreme observations. 
These "additional" extreme observations are assumed to arise from the contaminant BB-D, which exhibits higher dispersion in our mixture framework.
We place particular emphasis on the flexibility in capturing key characteristics of interest, including mean, variance, skewness, and kurtosis. 
Another notable advantage of the cBB-D is that its formulation introduces only two additional (contamination) parameters with respect to the reference BB-D, both of which have intuitive and practical interpretations. 
These parameters represent the proportion of observations from the contaminant BB-D and the degree of contamination, where the latter quantifies the extent to which the contaminant BB-D is more dispersed than the reference BB-D.

Another strength of the cBB-D is its parameterization, which includes a mean-related parameter and a dispersion parameter in addition to the contamination parameters. 
This structure allows the cBB-D to be seamlessly integrated into a regression framework, giving rise to the cBB regression model (cBB-RM). 
Notably, this regression framework is not restricted to modeling the mean only; rather, it extends to all parameters of the cBB-D, enabling the inclusion of different covariates for each parameter. 
This results in a highly flexible regression model for bounded count response variables, particularly in the presence of an excess of extreme values.

From an inferential standpoint, we propose an expectation-maximization (EM) algorithm for the maximum likelihood estimation of cBB-RM parameters. 
The robustness of the model to outliers is investigated via a sensitivity analysis, focusing on their potential to introduce bias in the estimated regression parameters. 
The real-world applications in the context of climate change further validate the effectiveness of our approach, demonstrating its advantages over established regression models for bounded counts. 
These results position the cBB-RM as a strong alternative for analyzing bounded count data characterized by an excess of extreme observations.


\section*{Funding acknowledgements}
Ferreira and Bekker have been partially supported by the National Research Foundation (NRF) of South Africa (SA), grant RA201125576565, nr 145681; RA171022270376, grant Nr. 119109; and grant SRUG2204203865 nr. 120839. 
The opinions expressed and conclusions arrived at are those of the authors and are not necessarily to be attributed to the NRF. Tortora has been partially supported by NSF grant Nr. 2209974.
\section*{Data availability statement}
All datasets considered in this paper are freely available on the internet.

\section*{Disclosure statement}
The authors declared no potential conflicts of interest with respect to the research, authorship, and/or publication of this article.
\bibliographystyle{chicago}
\bibliography{database.bib}

\begin{thebibliography}{}

\bibitem[\protect\citeauthoryear{Akaike}{Akaike}{1974}]{akaike1974new}
Akaike, H. (1974).
\newblock A new look at the statistical model identification.
\newblock {\em {IEEE} {T}ransactions on {A}utomatic {C}ontrol\/}~{\em 19\/}(6),
  716--723.

\bibitem[\protect\citeauthoryear{Ar{\i}kan and G{\"u}nay}{Ar{\i}kan and
  G{\"u}nay}{2021}]{arikan2021public}
Ar{\i}kan, G. and D.~G{\"u}nay (2021).
\newblock Public attitudes towards climate change: {A} cross-country analysis.
\newblock {\em The British Journal of Politics and International
  Relations\/}~{\em 23\/}(1), 158--174.

\bibitem[\protect\citeauthoryear{Arostegui, Padierna, and Quintana}{Arostegui
  et~al.}{2010}]{arostegui2010assessment}
Arostegui, I., A.~Padierna, and J.~M. Quintana (2010).
\newblock Assessment of {HRQ}o{L} in patients with eating disorders by the
  beta-binomial regression approach.
\newblock {\em International Journal of Eating Disorders\/}~{\em 43\/}(5),
  455--463.

\bibitem[\protect\citeauthoryear{Aspinall and Matthews}{Aspinall and
  Matthews}{1994}]{aspinall1994climate}
Aspinall, R. and K.~Matthews (1994).
\newblock Climate change impact on distribution and abundance of wildlife
  species: {A}n analytical approach using {GIS}.
\newblock {\em Environmental Pollution\/}~{\em 86\/}(2), 217--223.

\bibitem[\protect\citeauthoryear{Bayes, Baz{\'a}n, and Valdivieso}{Bayes
  et~al.}{2024}]{bayes2024robust}
Bayes, C.~L., J.~L. Baz{\'a}n, and L.~Valdivieso (2024).
\newblock A robust regression model for bounded count health data.
\newblock {\em Statistical Methods in Medical Research\/}, 09622802241259178.

\bibitem[\protect\citeauthoryear{Commission and European~Parliament}{Commission
  and European~Parliament}{2023}]{Eurobarometer}
Commission, E. and B.~European~Parliament (2023).
\newblock Eurobarometer 95.1 (2021).
\newblock GESIS, Cologne. ZA7781 Data file Version 2.0.0,
  https://doi.org/10.4232/1.14079.

\bibitem[\protect\citeauthoryear{Davies and Gather}{Davies and
  Gather}{1993}]{davies1993identification}
Davies, L. and U.~Gather (1993).
\newblock The identification of multiple outliers.
\newblock {\em Journal of the {A}merican {S}tatistical {A}ssociation\/}~{\em
  88\/}(423), 782--792.

\bibitem[\protect\citeauthoryear{Gallop, Rieger, McClintock, and Atkins}{Gallop
  et~al.}{2013}]{gallop2013model}
Gallop, R.~J., R.~H. Rieger, S.~McClintock, and D.~C. Atkins (2013).
\newblock A model for extreme stacking of data at endpoints of a distribution:
  {I}llustration with {W}-shaped data.
\newblock {\em Statistical Methodology\/}~{\em 10\/}(1), 29--45.

\bibitem[\protect\citeauthoryear{Griffiths}{Griffiths}{1973}]{griffiths1973maximum}
Griffiths, D. (1973).
\newblock Maximum likelihood estimation for the beta-binomial distribution and
  an application to the household distribution of the total number of cases of
  a disease.
\newblock {\em Biometrics\/}, 637--648.

\bibitem[\protect\citeauthoryear{Hannan and Quinn}{Hannan and
  Quinn}{1979}]{hannan1979determination}
Hannan, E.~J. and B.~G. Quinn (1979).
\newblock The determination of the order of an autoregression.
\newblock {\em Journal of the Royal Statistical Society: Series B
  (Methodological)\/}~{\em 41\/}(2), 190--195.

\bibitem[\protect\citeauthoryear{Keller-Ressel}{Keller-Ressel}{2022}]{keller2022w}
Keller-Ressel, M. (2022).
\newblock W-shaped implied volatility curves in a variance-gamma mixture model.
\newblock {\em arXiv preprint arXiv:2209.14726\/}.

\bibitem[\protect\citeauthoryear{Korkmaz}{Korkmaz}{2020}]{korkmaz2020new}
Korkmaz, M.~{\c{C}}. (2020).
\newblock A new heavy-tailed distribution defined on the bounded interval:
  {T}he logit slash distribution and its application.
\newblock {\em Journal of Applied Statistics\/}~{\em 47\/}(12), 2097--2119.

\bibitem[\protect\citeauthoryear{Ley, Babi{\'c}, and Craens}{Ley
  et~al.}{2021}]{ley2021flexible}
Ley, C., S.~Babi{\'c}, and D.~Craens (2021).
\newblock Flexible models for complex data with applications.
\newblock {\em Annual Review of Statistics and Its Application\/}~{\em 8\/}(1),
  369--391.

\bibitem[\protect\citeauthoryear{Martin, Witten, and Willis}{Martin
  et~al.}{2020}]{martin2020modeling}
Martin, B.~D., D.~Witten, and A.~D. Willis (2020).
\newblock Modeling microbial abundances and dysbiosis with beta-binomial
  regression.
\newblock {\em The Annals of Applied Statistics\/}~{\em 14\/}(1), 94.

\bibitem[\protect\citeauthoryear{Muluneh}{Muluneh}{2021}]{muluneh2021impact}
Muluneh, M.~G. (2021).
\newblock Impact of climate change on biodiversity and food security: {A}
  global perspective—a review article.
\newblock {\em Agriculture \& Food Security\/}~{\em 10\/}(1), 1--25.

\bibitem[\protect\citeauthoryear{Otto, Ferreira, Bekker, Punzo, and
  Tomarchio}{Otto et~al.}{2024}]{otto2024refreshing}
Otto, A., J.~Ferreira, A.~Bekker, A.~Punzo, and S.~Tomarchio (2024).
\newblock A refreshing take on the inverted {D}irichlet via a mode
  parameterization with some statistical illustrations.
\newblock {\em Journal of the {K}orean {S}tatistical Society\/}, 1--28.

\bibitem[\protect\citeauthoryear{Otto, Ferreira, Tomarchio, Bekker, and
  Punzo}{Otto et~al.}{2025}]{otto2025contaminated}
Otto, A.~F., J.~T. Ferreira, S.~D. Tomarchio, A.~Bekker, and A.~Punzo (2025).
\newblock A contaminated regression model for count health data.
\newblock {\em Statistical Methods in Medical Research\/}, 09622802241307613.

\bibitem[\protect\citeauthoryear{Paul and Saha}{Paul and
  Saha}{2007}]{paul2007generalized}
Paul, S. and K.~K. Saha (2007).
\newblock The generalized linear model and extensions: {A} review and some
  biological and environmental applications.
\newblock {\em Environmetrics\/}~{\em 18\/}(4), 421--443.

\bibitem[\protect\citeauthoryear{Punzo}{Punzo}{2019}]{Punz:Anew:2019}
Punzo, A. (2019).
\newblock A new look at the inverse {G}aussian distribution with applications
  to insurance and economic data.
\newblock {\em Journal of Applied Statistics\/}~{\em 46\/}(7), 1260--1287.

\bibitem[\protect\citeauthoryear{Punzo and Bagnato}{Punzo and
  Bagnato}{2021}]{Punz:Bagn:PhyA:2021}
Punzo, A. and L.~Bagnato (2021).
\newblock Modeling the cryptocurrency return distribution via {L}aplace scale
  mixtures.
\newblock {\em Physica A: Statistical Mechanics and its Applications\/}~{\em
  563\/}(1), 125354.

\bibitem[\protect\citeauthoryear{Punzo and Bagnato}{Punzo and
  Bagnato}{2024}]{punzo2024asymmetric}
Punzo, A. and L.~Bagnato (2024).
\newblock Asymmetric {L}aplace scale mixtures for the distribution of
  cryptocurrency returns.
\newblock {\em Advances in Data Analysis and Classification\/}, 1--48.

\bibitem[\protect\citeauthoryear{Ryan}{Ryan}{2007}]{ryan2007application}
Ryan, D.~A. (2007).
\newblock Application of the beta-binomial model for the detection of rare
  marine benthos using point intercept techniques.
\newblock {\em Environmetrics\/}~{\em 18\/}(4), 361--373.

\bibitem[\protect\citeauthoryear{Schwarz}{Schwarz}{1978}]{schwarz1978estimating}
Schwarz, G. (1978).
\newblock Estimating the dimension of a model.
\newblock {\em The {A}nnals of {S}tatistics\/}~{\em 6\/}(2), 461--464.

\bibitem[\protect\citeauthoryear{Sciandra, Fasola, Albano, Di~Maria, and
  Plaia}{Sciandra et~al.}{2024}]{sciandra2024discrete}
Sciandra, M., S.~Fasola, A.~Albano, C.~Di~Maria, and A.~Plaia (2024).
\newblock Discrete beta and shifted beta-binomial models for rating and ranking
  data.
\newblock {\em Environmental and Ecological Statistics\/}, 1--22.

\bibitem[\protect\citeauthoryear{Tomarchio, Punzo, Ferreira, and
  Bekker}{Tomarchio et~al.}{2025}]{tomarchio2024new}
Tomarchio, S.~D., A.~Punzo, J.~T. Ferreira, and A.~Bekker (2025).
\newblock A new look at the {D}irichlet distribution: robustness, clustering,
  and both together.
\newblock {\em Journal of Classification\/}~{\em 42}, 31--53.

\bibitem[\protect\citeauthoryear{Unsworth, Pac, White, and Bartmann}{Unsworth
  et~al.}{1999}]{unsworth1999mule}
Unsworth, J.~W., D.~F. Pac, G.~C. White, and R.~M. Bartmann (1999).
\newblock Mule deer survival in {C}olorado, {I}daho, and {M}ontana.
\newblock {\em The Journal of Wildlife Management\/}, 315--326.

\bibitem[\protect\citeauthoryear{Wagener, Bekker, Arashi, and Punzo}{Wagener
  et~al.}{2024}]{wagener2024uncovering}
Wagener, M., A.~Bekker, M.~Arashi, and A.~Punzo (2024).
\newblock Uncovering a generalised gamma distribution: {F}rom shape to
  interpretation.
\newblock {\em Results in Applied Mathematics\/}~{\em 22}, 100461.

\bibitem[\protect\citeauthoryear{Yee, Santavy, and Barron}{Yee
  et~al.}{2008}]{yee2008comparing}
Yee, S.~H., D.~L. Santavy, and M.~G. Barron (2008).
\newblock Comparing environmental influences on coral bleaching across and
  within species using clustered binomial regression.
\newblock {\em Ecological Modelling\/}~{\em 218\/}(1-2), 162--174.

\end{thebibliography}

\newpage

\appendix 

\section{Proofs} 
\label{ProofinApp}

\renewcommand{\thefigure}{A.\arabic{figure}}
\setcounter{figure}{0}

\noindent\emph{Proof Proposition 1:}

The cBB distribution in \eqref{pdf contaminated beta binomial} has the hierarchical representation
\begin{eqnarray}
    W   &\sim& \mathcal{TP}_{\{1,\eta\}}(\delta)\nonumber\\
    Y|W=w&\sim&\mathcal{BB}_m(\pi,w\sigma), \label{eq:hirc2}
\end{eqnarray}
where $\mathcal{TP}_{\{1,\eta\}}(\delta)$ denotes a two-point random variable with probability of success $\delta$ on the support $\left\{1,\eta\right\}$ defined as
\begin{equation}\label{bernoulli dist}
W=\begin{cases}
    1 &\text{with probability $1-\delta$,} \\
    \eta &\text{with probability $\delta$.} 
\end{cases}
\end{equation}
The proofs of \ref{item:prop1}--\ref{item:prop4} in Proposition \ref{proposition 1} follow:
\begin{enumerate}[label=(\itshape\alph*\upshape)]
\item 
\label{P1Proof} 
if $\delta\to0^{+}$, from \eqref{bernoulli dist} it follows that $W \overset{D}{\to} 1$ and, therefore, according to \eqref{eq:hirc2}--\eqref{bernoulli dist}, $Y\overset{D}{\to}\mathcal{BB}_m(\pi,\sigma)$;  	
      \item 
      \label{P2Proof} 
      if $\eta\to1^{+}$, from \eqref{bernoulli dist} it follows that $W \overset{D}{\to} 1$ and, as before, according to \eqref{eq:hirc2}--\eqref{bernoulli dist}, $Y\overset{D}{\to}\mathcal{BB}_m(\pi,\sigma)$; 
      \item if $\delta\to0^{+}$ and $\sigma\to0^{+}$, from the proof for \ref{P1Proof} and from the results given in \citet{griffiths1973maximum}, it follows that $Y\overset{D}{\to}\mathcal{B}_m(\pi)$; \item if $\eta\to1^{+}$ and $\sigma\to0^{+}$, from the proof for \ref{P2Proof} and, again, as demonstrated in \citet{griffiths1973maximum}, it follows that $Y\overset{D}{\to}\mathcal{B}_m(\pi)$.
\end{enumerate}

\begin{prop}
     Characteristics of the cBB distribution.
If $Y\sim c\mathcal{BB}_m(\pi,\sigma,\delta,\eta)$, then \begin{enumerate}[label=(\itshape\alph*\upshape)]
    \item the variance is given by \begin{align*}
        \text{Var}_{\text{cBB}_m}(Y;\pi,\sigma,\delta,\eta)
        &=\frac{m\pi(1-\pi)\left[(1-\delta)(1+m\sigma)(1+\eta\sigma)+\delta(1+m\eta\sigma)(1+\sigma)\right]}{(1+\sigma)(1+\eta\sigma)};
    \end{align*}
    \item the skewness is given by 
    \begin{align*}
    \text{Skew}_{\text{cBB}_m}(Y;\pi,\sigma,\delta,\eta)&=\frac{1-\delta}{\left(\text{Var}_{\text{cBB}_m}(Y;\pi,\sigma,\delta,\eta)\right)^\frac{3}{2}}\left(\frac{m\pi(1-\pi)(1-2\pi)(1+m\sigma )(1+2m\sigma )}{(1+\sigma )(1+2\sigma )}\right)\\
          &+
          \frac{\delta}{\left(\text{Var}_{\text{cBB}_m}(Y;\pi,\sigma,\delta,\eta)\right)^{3}{2}}\left(\frac{m\pi(1-\pi)(1-2\pi)(1+m\sigma \eta)(1+2m\sigma \eta)}{(1+\sigma \eta)(1+2\sigma \eta)}\right);
          \end{align*}
     \item the kurtosis is given by  \begin{align*}
    &\text{ExKurt}_{\text{cBB}_m}(Y;\pi,\sigma,\delta,\eta)\\
    &=-3+\frac{m \pi   (1-\pi ) }{\left[{\text{Var}_{\text{cBB}}(Y;\pi,\sigma,\delta,\eta)}\right]^2}\\
    &\times\left(\frac{(1-\delta)(m \sigma +1) \left(6 (3 (\pi -1) \pi +1) m^2 \sigma ^2+3 m \sigma  (2-(\pi -1) \pi  (m-6))-3 (\pi -1) \pi  (m-2)-\sigma +1\right)}{(\sigma +1) (2 \sigma +1) (3 \sigma +1)}\right.\\
    &+\left.\frac{\delta(\eta  m \sigma +1) \left(-\eta  \sigma +6 \eta ^2 (3 (\pi -1) \pi +1) m^2 \sigma ^2+3 \eta  m \sigma  (2-(\pi -1) \pi  (m-6))-3 (\pi -1) \pi  (m-2)+1\right)}{(\eta  \sigma +1) (2 \eta  \sigma +1) (3 \eta  \sigma +1)}\right)\\
\end{align*}
\end{enumerate}
\end{prop}
\begin{proof}
If $Y\sim c\mathcal{BB}_m(\pi,\sigma,\delta,\eta)$, then from the hierarchical representation in \eqref{eq:hirc2}--\eqref{bernoulli dist}, 
\begin{enumerate}[label=(\itshape\alph*\upshape)]
    \item \begin{align*}
        \text{Var}_{\text{cBB}_m}(Y;\pi,\sigma,\delta,\eta)&=\mathrm{E}_{\text{cBB}_m}\left(Y^2;\pi,\sigma,\delta,\eta\right)-\left[\mathrm{E}_{\text{cBB}_m}(Y;\pi,\sigma,\delta,\eta)\right]^2\\
 &=\mathrm{E}_{\text{TP}}\left[\text{Var}_{\text{BB}_m}(Y|W=w;\pi,\sigma);\delta\right]+\text{Var}_{\text{TP}}\left[\mathrm{E}_{\text{BB}_m}(Y|W=w;\pi,\sigma);\delta\right]\\
    \end{align*}
    where 
    \begin{align*}
        \text{Var}_{\text{BB}_m}(Y|W=w;\pi,\sigma)=\frac{m\pi(1-\pi)(1+mw\sigma)}{1+w\sigma}.
    \end{align*}
    Since $\mathrm{E}_{\text{BB}_m}(Y|W=w;\pi,\sigma)=m\pi$, it follows that
    \begin{align*}
        \text{Var}_{\text{cBB}_m}(Y;\pi,\sigma,\delta,\eta)&=\mathrm{E}_{\text{TP}}\left[\frac{m\pi(1-\pi)(1+mw\sigma)}{1+w\sigma};\delta\right]\\
        &=(1-\delta)\frac{m\pi(1-\pi)(1+m\sigma)}{1+\sigma}+\delta\frac{m\pi(1-\pi)(1+m\eta\sigma)}{1+\eta\sigma}\\
        &=\frac{m\pi(1-\pi)\left[(1-\delta)(1+m\sigma)(1+\eta\sigma)+\delta(1+m\eta\sigma)(1+\sigma)\right]}{(1+\sigma)(1+\eta\sigma)}.
    \end{align*}
    \item \begin{align*}
        \text{Skew}_{\text{cBB}_m}(Y;\pi,\sigma,\delta,\eta)&=\mathrm{E}_{\text{cBB}_m}
        \left[\left(\frac{Y-\mathrm{E}_{\text{cBB}_m}(Y;\pi,\sigma,\delta,\eta)}{\sqrt{\text{Var}_{\text{cBB}_m}(Y;\pi,\sigma,\delta,\eta)}}\right)^3;\pi,\sigma,\delta,\eta\right]\\
        &=\sum_{i=0}^m\left(\frac{y_i-\mathrm{E}_{\text{cBB}_m}(Y;\pi,\sigma,\delta,\eta)}{\sqrt{\text{Var}_{\text{cBB}_m}(Y;\pi,\sigma,\delta,\eta)}}\right)^3f_{\text{cBB}_m}(y_i;\pi,\sigma,\delta,\eta)\\
         &=\sum_w\sum_{i=0}^m\left(\frac{y_i-m\pi}{\sqrt{\text{Var}_{\text{cBB}_m}(Y;\pi,\sigma,\delta,\eta)}}\right)^3f_{\text{BB}_m}(y_i;\pi,\sigma w)\\
         &=\frac{1}{\left(\text{Var}_{\text{cBB}_m}(Y;\pi,\sigma,\delta,\eta)\right)^{\frac{3}{2}}}\sum_w\sum_{i=0}^m\left({y_i-m\pi}\right)^3f_{\text{BB}_m}(y_i;\pi,\sigma w)\\
          &=\frac{1}{\left(\text{Var}_{\text{cBB}_m}(Y;\pi,\sigma,\delta,\eta)\right)^\frac{3}{2}}\sum_w\left(\frac{m\pi(1-\pi)(2\pi-1)(1+m\sigma w)(1+2m\sigma w)}{(1+\sigma w)(1+2\sigma w)}\right)\\
          &=\frac{1-\delta}{\left(\text{Var}_{\text{cBB}_m}(Y;\pi,\sigma,\delta,\eta)\right)^\frac{3}{2}}\left(\frac{m\pi(1-\pi)(1-2\pi)(1+m\sigma )(1+2m\sigma )}{(1+\sigma )(1+2\sigma )}\right)\\
          &+
          \frac{\delta}{\left(\text{Var}_{\text{cBB}_m}(Y;\pi,\sigma,\delta,\eta)\right)^{3}{2}}\left(\frac{m\pi(1-\pi)(1-2\pi)(1+m\sigma \eta)(1+2m\sigma \eta)}{(1+\sigma \eta)(1+2\sigma \eta)}\right).
    \end{align*}
    
\item    \begin{align*}
        \text{Kurt}_{\text{cBB}_m}(Y;\pi,\sigma,\delta,\eta)&=\mathrm{E}_{\text{cBB}_m}
        \left[\left(\frac{Y-\mathrm{E}_{\text{cBB}_m}(Y;\pi,\sigma,\delta,\eta)}{\sqrt{\text{Var}_{\text{cBB}_m}(Y;\pi,\sigma,\delta,\eta)}}\right)^4;\pi,\sigma,\delta,\eta\right]\\
        &=\sum_{i=0}^m\left(\frac{y_i-\mathrm{E}_{\text{cBB}_m}(Y;\pi,\sigma,\delta,\eta)}{\sqrt{\text{Var}_{\text{cBB}_m}(Y;\pi,\sigma,\delta,\eta)}}\right)^4f_{\text{cBB}_m}(y_i;\pi,\sigma,\delta,\eta)\\
         &=\sum_w\sum_{i=0}^m\left(\frac{y_i-m\pi}{\sqrt{\text{Var}_{\text{cBB}_m}(Y;\pi,\sigma,\delta,\eta)}}\right)^4f_{\text{BB}_m}(y_i;\pi,\sigma w)\\
         &=\sum_w\frac{1}{\left[{\text{Var}_{\text{cBB}_m}(Y;\pi,\sigma,\delta,\eta)}\right]^2}\sum_{i=0}^m\left({y_i-m\pi}\right)^4f_{\text{BB}_m}(y_i;\pi,\sigma w)\\
    \end{align*}
    where
    \begin{align*}
        &\sum_{i=0}^m\left({y_i-m\pi}\right)^4f_{\text{BB}_m}(y_i;\pi,\sigma w)\\
        &=\frac{m \pi  (1-\pi)  (m \sigma  w+1) \left(6 (3 (\pi -1) \pi +1) m^2 \sigma ^2 w^2-3 (\pi -1) \pi  (m-2)+3 m \sigma  w (2-(\pi -1) \pi  (m-6))-\sigma  w+1\right)}{(\sigma  w+1) (2 \sigma  w+1) (3 \sigma  w+1)}.
    \end{align*}
    Thus,
\begin{align*}
    &\text{ExKurt}_{\text{cBB}_m}(Y;\pi,\sigma,\delta,\eta)\\
    &=-3+\frac{m \pi   (1-\pi ) }{\left[{\text{Var}_{\text{cBB}_m}(Y;\pi,\sigma,\delta,\eta)}\right]^2}\\
    &\times\left(\frac{(1-\delta)(m \sigma +1) \left(6 (3 (\pi -1) \pi +1) m^2 \sigma ^2+3 m \sigma  (2-(\pi -1) \pi  (m-6))-3 (\pi -1) \pi  (m-2)-\sigma +1\right)}{(\sigma +1) (2 \sigma +1) (3 \sigma +1)}\right.\\
    &+\left.\frac{\delta(\eta  m \sigma +1) \left(-\eta  \sigma +6 \eta ^2 (3 (\pi -1) \pi +1) m^2 \sigma ^2+3 \eta  m \sigma  (2-(\pi -1) \pi  (m-6))-3 (\pi -1) \pi  (m-2)+1\right)}{(\eta  \sigma +1) (2 \eta  \sigma +1) (3 \eta  \sigma +1)}\right).\\
\end{align*}
\end{enumerate}
\end{proof}

\end{document}